\documentclass[11pt, draftcls, onecolumn]{IEEEtran}

\def\withcolors{0}
\def\withnotes{0}

\input{packages}
\input{preamble}

\newcommand{\Rappor}{\ensuremath{\textsc{RAPPOR}}\xspace}
\newcommand{\alprappor}{\alpha}
\newcommand{\betrappor}{\beta}

\begin{document}
\title{Inference under Information Constraints III: Local Privacy Constraints\footnote{A preliminary version of this work containing partial results appeared in the \emph{Proceedings of the 22\textsuperscript{nd} International Conference on Artificial Intelligence and Statistics (AISTATS)}, 2019~\cite{ACFT:19}. }}
\date{}

\author{
  \IEEEauthorblockN{Jayadev Acharya\IEEEauthorrefmark{1}} 
  \and \IEEEauthorblockN{Cl\'{e}ment L. Canonne\IEEEauthorrefmark{2}}
  \and \IEEEauthorblockN{Cody Freitag\IEEEauthorrefmark{3}}
  \and \IEEEauthorblockN{Ziteng Sun\IEEEauthorrefmark{4}} 
  \and \IEEEauthorblockN{Himanshu Tyagi\IEEEauthorrefmark{5}}
}

\maketitle 

{
\renewcommand{\thefootnote}{}
  \footnotetext{%
      \IEEEauthorblockA{\IEEEauthorrefmark{1}Cornell University. Email: acharya@cornell.edu. Supported by NSF-CCF-1846300 (CAREER), NSF-CCF-1815893, and a Google Faculty Research Award.}\\
      \indent\IEEEauthorblockA{\IEEEauthorrefmark{2}University of Sydney. Email: ccanonne@cs.columbia.edu. This work was performed while a Goldstine Postdoctoral Fellow at IBM Research, and a Motwani Postdoctoral Fellow at Stanford University.}\\
      \indent\IEEEauthorblockA{\IEEEauthorrefmark{3}Cornell Tech. Email: cfreitag@cs.cornell.edu. Supported in part by NSF GRFP award DGE-1650441.}\\
      \indent\IEEEauthorblockA{\IEEEauthorrefmark{4}Cornell University. Email: zs335@cornell.edu. Supported in part by NSF-CCF-1846300 (CAREER).}\\
      \indent\IEEEauthorblockA{\IEEEauthorrefmark{5}Indian Institute of Science.  Email: htyagi@iisc.ac.in.
Supported in part by a research grant from the Robert Bosch Center for Cyberphysical Systems (RBCCPS), Indian Institute of Science, Bangalore.  
  }
}
\renewcommand{\thefootnote}{\arabic{footnote}}
\setcounter{footnote}{0}

\begin{abstract}

We study goodness-of-fit and independence testing of discrete distributions in a \new{setting} where samples are distributed across multiple users. The users wish to %
preserve the privacy of their \new{data} while enabling a central server to perform the tests. Under the notion of local differential privacy, we propose simple, sample-optimal, and communication-efficient protocols for these two questions in the noninteractive setting,  where in addition users may or may not share a common random seed. In particular, we show that the availability of shared (public) randomness greatly reduces the sample complexity. 
Underlying our public-coin protocols are privacy-preserving mappings which, when applied to the samples, minimally contract the distance between their respective probability distributions.

\end{abstract}
\newpage

\section{Introduction}\label{sec:intro}
  Inferring statistical properties of data sources while maintaining
their privacy is a core problem in privacy-preserving statistics. A
widely established notion to achieve this is \emph{local differential
privacy (LDP)}, introduced in~\cite{KLNRS:08, EvfimievskiGS:03}. The
data samples are distributed across users (``players''), who do not
trust the centralized data curator, which can be \eg{} corporate
entities or government agencies. The data samples are privatized via a
noise addition mechanism that is locally differentially private
(see~\cref{def:ldp:channel}). This falls under the general setting of
statistical inference under \emph{local information constraints},
namely constraints on information that each player can reveal about
its sample.

Recently, a subset of the authors have initiated a systematic 
study of such problems under general constraints. 
In particular,~\cite{AcharyaCT:IT1} provides
a framework for deriving lower bounds for such problems
and~\cite{ACT:19} provides sample-optimal
algorithms for communication constraints. This paper, the third in
this series, focuses on local privacy constraints. 
Specifically, we consider two of the most
fundamental goodness-of-fit tasks, testing identity and independence
of discrete distributions, and design sample-optimal LDP mechanisms
for these tasks. We
restrict to simultaneous message passing protocols and 
lay special emphasis on the availability
of \emph{public randomness} at the players (\ie, a common random seed
shared by all parties)\footnote{We assume private randomness is always available at the players. Formal definitions can be found in Section~\ref{sec:setup}.} and seek to answer the following.

\begin{center}
\textit{What is the sample complexity of testing identity and independence of discrete distributions under local differential privacy? Does the sample complexity depend on whether public randomness is available?}
\end{center}
The role of public randomness in the design and analysis of
distributed statistical inference has hitherto been largely
overlooked. We fully resolve this question by providing tight bounds
on the sample complexity of identity testing and independence testing
of discrete distributions under local differential privacy, both with
and without public randomness. Our results show that, for these two
composite hypothesis testing tasks, schemes which allow for public
randomness can achieve significantly smaller sample complexity than
those who do not.  Interestingly, this is in contrast with the seminal
work of Tsitsiklis~\cite{Tsitsiklis:93}, which established that public
randomness provides no advantage in the context of
distributed \emph{simple} hypothesis testing without local privacy constraints.

\subsection{Results and techniques}
We study two inference problems over discrete distributions, identity
testing and independence testing under $\priv$-LDP
(at a high level, the privacy parameter $\priv>0$ bounds the (worst-case) statistical leakage of any player's data, and smaller values imply stronger privacy guarantees; see~\cref{sec:setup} for formal definitions). Our results are
summarized in~\cref{table:results}; we outline and discuss them below.

\begin{table}[ht]\centering
\renewcommand{\arraystretch}{1.25}
\caption{Summary of our results and previous work.}\label{table:results}
\begin{tabular}{|c|>{\columncolor{red!10}} c|>{\columncolor{blue!10}} c|>{\columncolor{red!10}} c|>{\columncolor{blue!10}} c| }\hline
\multirow{2}{*}{} & \multicolumn{2}{|c|}{\textbf{This work}} & \multicolumn{2}{|c|}{\textbf{Previous work}} \\\cline{2-5}
& Private-Coin & Public-Coin & Private-Coin & Public-Coin \\\hline
Identity Testing & $O\Paren{\frac{\ab^{3/2}}{\dst^2\priv^2} }$ & $\bigO{\frac{\ab}{\dst^2\priv^2} }$  & $\bigO{\frac{\ab^{2}}{\dst^2\priv^2} }$, $\bigOmega{\frac{\ab^{3/2}}{\dst^2\priv^2} }$ & $\bigOmega{\frac{\ab}{\dst^2\priv^2} }$ \\\hline
Independence Testing & $\bigTheta{\frac{\ab^{3}}{\dst^2\priv^4} }$  & $\bigTheta{\frac{\ab^{2}}{\dst^2\priv^2} }$  & \multicolumn{2}{|c|}{\cellcolor{red!10}$O\Paren{\frac{\ab^{4}}{\dst^2\priv^2} }$ } \\\hline
\end{tabular}
\end{table}

In the identity testing question, there is a known reference
distribution $\q$ over $[\ab]\eqdef\{1,\ldots, \ab\}$, and the
players' samples are i.i.d. from an unknown distribution $\p$. The
goal is to test the hypotheses $\mathcal{H}_0: \p=\q$ and
$\mathcal{H}_1:\totalvardist{\p}{\q}>\dst$ using $\priv$-LDP
mechanisms.  We seek to characterize the \emph{sample complexity}
of this task, which is the minimum number of players to solve this
problem with a (small) constant two-sided error. Without privacy
constraints, when the true samples of $\p$ are available to the
central data curator (``referee''), the optimal sample complexity of
identity testing is known to be $\Theta(\ab^{1/2}/\dst^2$).

There are two parts of the problem. The first is to design \emph{privacy-preserving mechanisms} that the players use to encode their data to be sent the server. 
The second is to design \emph{post-processing algorithms} that the server uses to decide the output of the test given the privatized messages. 

We first consider the task of designing optimal post-processing algorithms for existing $\rho$-LDP mechanisms. This is of interest in cases where the privatization mechanisms are in place, and changing them is impossible or too expensive---for instance, when an organization has already deployed a data aggregation pipeline, and seeks to add a statistical inference component to it without overhauling the entire system.

Arguably the simplest privatization scheme is $\ab$-randomized response (see~\cite{Warner:65}). Unfortunately, it was shown in~\cite{Sheffet:18} that the sample complexity of any test relying on this scheme is $\Theta(\ab^{5/2}/\dst^2\priv^2)$, far from optimal. Our first result considers the now well established privatization scheme \Rappor (Randomized Aggregatable Privacy-Preserving Ordinal Response~\cite{ErlingssonPK14, KairouzBR16}) (see~\ref{ssec:identity:rappor}). In~\cref{theo:identity:private:rappor}, we design an identity testing algorithm that, given samples from the \Rappor mechanism, has sample complexity $\bigO{\frac{\ab^{3/2}}{\dst^2\priv^2}}$---a factor $\ab$ improvement over randomized response. 

The \Rappor mechanism produces privatized messages with $\Omega(\ab)$ bits of entropy, and as a result those messages are $\ab$-bit long. Thus, \Rappor requires a large communication bandwidth. We provide a new mechanism based on the recently proposed Hadamard Response (HR) that produces only one-bit messages, leading to an identity testing algorithm with the same sample complexity $\bigO{\frac{\ab^{3/2}}{\dst^2\priv^2}}$. This result is given in~\cref{theo:identity:private:onebit}. 

All the schemes above require no publicly agreed upon randomness, which we refer to as \emph{private-coin} mechanisms, and are highly desirable when it is too inefficient or infeasible to setup a common random seed.
However, in~\cite{AcharyaCT:IT1}, it was established that any testing algorithm based on any private-coin $\priv$-LDP mechanism must use $\Omega\left(\frac{\ab^{3/2}}{\dst^2\priv^2}\right)$ players. Therefore, the algorithms we propose based on \Rappor and HR are the best possible, and more significantly, are optimal among all LDP schemes that do not use public randomness. 

This raises the question of building LDP mechanisms that do use public randomness, which we refer to as \emph{public-coin} mechanisms, and post-processing algorithms that require fewer samples.
We emphasize that the public randomness is used only for added utility and we require the same strong privacy guarantees.
In this context, we design a new public-coin $\priv$-LDP mechanism and a corresponding algorithm whose sample complexity is $\bigO{\frac{\ab}{\dst^2\priv^2}}$, a factor $\sqrt{\ab}$ improvement over the best possible without using public randomness. Furthermore, this is asymptotically optimal from the result of~\cite{AcharyaCT:IT1}, and the mechanism only uses one bit of communication from each player, making it as communication-efficient as possible. Our result relies on a randomized one-bit isometry, where the players use the common random seed to randomly project the original domain $[\ab]$ to a binary domain and perform testing over this new domain. This result is given in~\cref{theo:identity:public:onebit:restated}.

We then turn to the task of independence testing. Here, the underlying distribution $\p$ is over the product domain $[\ab]\times[\ab]$, and the goal is to test whether the marginals of $\p$ are independent (\ie{} if $\p$ is a product distribution) or at least $\dst$ away from all product distributions. We design schemes without and with public randomness which achieve sample complexity $\bigO{\frac{\ab^{3}}{\dst^2\priv^2}}$ and $\bigO{\frac{\ab^{2}}{\dst^2\priv^2}}$, respectively. These results are given in~\cref{theo:independence:private} and~\cref{theo:independence:public:onebit:restated}. Interestingly, in the case where public randomness is available, our protocol relies on a one-bit isometry similar to the one used in the identity testing case, but suitably generalized to handle the product structure of the domain. 
Finally, we prove the optimality of both these bounds, establishing matching lower bounds in~\cref{corollary:independence:uniformity:reduction}. This is done by providing a formal reduction from independence testing over $[\ab]\times[\ab]$ to the identity testing problem over $[\ab^2]$. We believe this general reduction, which is not specific to the locally private setting, to be of independent interest.

The conceptual takeaway message of our results is that, for composite hypothesis testing problems, public randomness can prove very helpful, and its availability leads to significantly more sample-efficient protocols. 

We finally remark that although this work is concerned with noninteractive protocols, more complicated \emph{adaptive} LDP schemes are possible where the players sequentially choose their privatization schemes upon observing the messages of all previous players and the available public randomness. Recent works in this setting~\cite{BB:20, AJM:20, ACLST:20} show that, for identity testing, adaptivity does not allow for more efficient protocols than public randomness, and by our reduction for independence testing this carries over to the independence testing problem as well. 

\subsection{Related prior work}
Testing properties of distributions from their samples has a long history in statistics, which dates back more than a century. Recently, this problem has gathered renewed interest in the computer science community, with a particular focus on the study of discrete distributions in the finite-sample regime. In this section, we only focus on closely related papers and we refer an interested reader to surveys and books~\cite{Rubinfeld:12,Canonne:15,Goldreich:17,BW:17} for a comprehensive treatment. 

Following a long line of work, the optimal sample complexity for identity testing has been established as $\Theta(\ab^{1/2}/\dst^2)$~\cite{Paninski:08, Goldreich:16, ValiantV17a} under constant error probability. \cite{HuangM13,	DGPP:18} establish the optimal dependence on the error probability. \cite{ValiantV17a,BCG:19} also study the ``instance-optimal'' variant of the problem, introduced in~\cite{ValiantV17a}. The optimal sample complexity for the independence testing problem where both observations are from the same set $[\ab]$ was studied in~\cite{BFFKRW:01,LRR:13}, and shown to be $\Theta(\ab/\dst^2)$ in~\cite{AcharyaDK15,DK:16}.

Distribution testing has also been studied under privacy constraints on the samples. Under the notion of (global) {\em differential privacy} (DP)~\cite{DMNS:06}, 
identity testing has been considered in~\cite{CDK:17,
	ADR:17}, with a complete characterization of the sample complexity
derived in~\cite{ASZ:18:DP}. \cite{CKMUZ:19} focuses on the class of product distributions in high dimensions, including product of Bernoulli's and Gaussians with known variances. Both these works show that, in certain parameter regimes, the sample complexity can match the sample complexity of the non-private counterpart of the problem, which is in sharp contrast to the more stringent case of \emph{local} privacy (LDP) considered in this paper. Finally, \cite{AJM:20} and~\cite{BCJM:20} consider uniformity testing (a specific case of identity testing) under the notions of {\em pan-privacy} and \emph{shuffle privacy}, respectively, which provide privacy guarantees in-between DP and LDP. 

Independence testing under {\em differentially privacy} has been studied in~\cite{GaboardiLRV:16, KiferR:17, WangLK:15} and the first algorithm with finite sample guarantee was given in~\cite{ADKR:19}.

The works most closely related to ours are those that consider distribution testing under LDP constraints~\cite{Sheffet:18, GR:18, ACT:19, ACFT:19, ACHST:20, AJM:20, ACLST:20, BB:20}. \cite{Sheffet:18} considers both identity testing and independence testing with private-coin, noninteractive schemes. Our results improve upon theirs by a factor of $\ab$ and $\ab^2$, respectively. \cite{ACT:19} establishes lower bounds for identity testing using both private-coin and public-coin noninteractive schemes, which match our bounds in both cases and imply the optimality of our results. \cite{ACHST:20} considers noninteractive schemes where only a limited amount of public randomness is available, and obtains the optimal sample complexity which interpolates smoothly between the private-coin and public-coin cases. \cite{AJM:20, ACLST:20, BB:20} consider identity testing using sequentially interactive schemes, which combined with our results prove that interactivity cannot lead to an improvement in the sample complexity over public-coin noninteractive schemes. We note that the recent work of Joseph et al.~\cite{JMNR:19} also considers the role of interactivity in LDP hypothesis testing; however, they focus on simple hypothesis testing (as well as a generalization to \emph{convex} hypothesis classes). Their results do not apply to identity testing, and are incomparable to ours.

Another class of problems of statistical inference, density estimation, requires learning
the unknown distribution up to a desired accuracy of $\dst$ in total
variation distance. The optimal sample complexity of 
locally private learning discrete $\ab$-ary distributions is known to
be $\Theta(\ab^2/(\dst^2\priv^2))$; see~\cite{DJW:13, ErlingssonPK14, YeB17, KairouzBR16,
	ASZ:18:HR, AcharyaS:19}. The private-coin identity testing schemes in this paper are based on the same LDP randomization schemes proposed in these papers at the user side. Specifically, \Rappor was independently proposed in~\cite{ErlingssonPK14, DJW:13} and analyzed in~\cite{KairouzBR16}. Hadamard Response and its one-bit variant are proposed in~\cite{ASZ:18:HR, AcharyaS:19}. 
Also,
\cite{BNST:17} uses Hadamard transform together with sampling to reduce user communication to $O(1)$ bits in a public-coin scheme. Moreover, our private-coin independence testing protocol also involves a step that learns both marginal distributions, which relies on the scheme from~\cite{AcharyaS:19}.

\subsection{Organization}
The rest of the paper is organized as follows. In~\cref{ssec:identity:rappor,ssec:identity:hadamard} we provide two private-coin LDP schemes for identity testing based on \Rappor and Hadamard Response respectively, and analyze their sample complexity. In~\cref{sec:identity-public} we establish an upper bound on the  sample complexity of public-coin protocols for identity testing. In~\cref{sec:independence-private,sec:idependence-public} we establish the upper bounds on private-coin and public-coin independence testing, respectively. Finally, in~\cref{sec:lower-independence} we provide a reduction between identity and independence testing and use it to prove the optimality of the proposed independence tests both for private- and public-coin protocols.

\section{The setup: local privacy and inference protocols}\label{sec:setup}
\subsection{Notation}
Throughout the paper, we denote by  $\log$ the natural logarithm and $\log_2$ the base $2$ logarithm. We use standard asymptotic
notation $\bigO{\cdot}$, $\bigOmega{\cdot}$, and $\bigTheta{\cdot}$ for complexity orders.\footnote{Namely, for two non-negative sequences $(a_n)_n$ and $(b_n)_n$, we write $a_n = O(b_n)$ (resp., $a_n = \Omega(b_n)$) if there exist $C>0$ and $N\geq 0$ such that $a_n \leq Cb_n$ (resp., $a_n \geq Cb_n$) for all $n\geq N$. Further, we write $a_n = \Theta(b_n)$ when both $a_n = O(b_n)$ and $a_n = \Omega(b_n)$ hold.} 
 
For a known and fixed discrete domain $\domain$ let $\distribs{\domain}$ be the set of probability distributions over $\domain$, \ie,
\[ 
\distribs{\domain} = \setOfSuchThat{ \p\colon\domain\to[0,1] }{ \normone{\p}=1 },  
\]
where we identify a probability distribution to its probability mass function. We denote by $\uniformOn{\domain}$ the
uniform distribution on $\domain$ and omit the subscript when the
domain is clear from context. 

We are mostly interested in $\ab$-ary discrete distributions, and assume without loss of generality that $\domain=[\ab]\eqdef \{1,2,\dots,\ab\}$. We use $\distribs{[\ab]}$ and $\distribs{\ab}$ interchangeably to denote the probability simplex consisting of all distributions over $[\ab]$.

The \emph{total variation distance} between distributions $\p,\q\in\distribs{\domain}$ is 
\begin{equation*}
\totalvardist{\p}{\q} \eqdef \sup_{S\subseteq\domain} \left(\p(S)-\q(S)\right)
= \frac{1}{2} \sum_{x\in\domain} \abs{\p(x)-\q(x)},
\end{equation*}
namely, $\totalvardist{\p}{\q}$ is equal to half of the $\lp[1]$ distance
of $\p$ and $\q$.  For a distance parameter $\dst\in(0,1]$,
  we say that $\p,\q\in\distribs{\domain}$ are \emph{$\dst$-far} if  $\totalvardist{\p}{\q}>\dst$.
  Finally, for two distributions $\p_1$ and $\p_2$ over $\domain$, we denote by $\p_1\otimes\p_2$ the product
  distribution over $\domain\times\domain$ defined by
  $(\p_1\otimes\p_2)(x_1,x_2)=\p_1(x_1) \cdot \p_2(x_2)$ for all $x_1,x_2\in\domain$.

\subsection{Local differential privacy and protocols}

For a data domain $\domain$ and some set $\cY$ (which denotes the message set), a channel $W\colon\cX\to\cY$ is \emph{$\priv$-locally differentially private}
($\priv$-LDP) \emph{mechanism} if~\cite{EvfimievskiGS:03,DMNS:06,KLNRS:11}
\begin{equation}
  \label{def:ldp:channel}
    \max_{y\in \cY}\max_{x,x'\in \cX}\frac{W(y\mid x')}{W(y\mid x)} \leq e^{\priv}\,.   
\end{equation}
where, slightly overloading notation, we write $W(\cdot \mid x)$ for the output distribution (on $\cY$) for input $x\in\cX$. Loosely speaking, no output message from a user can reveal too much about their sample. Let $\cW_\priv$ be the set of all $\priv$-LDP channels with output $\{0,1\}^\ast$ the set of all binary strings. 

Our setup is depicted in \cref{fig:model}. There are $\ns$ independent samples $X^\ns\eqdef X_1,\ldots, X_{\ns}$ from an unknown distribution $\p$ distributed across $\ns$ players, with player $i$ holding $X_i$. Player $i$ passes $X_i$ through a privatization channel $W_i\in\cW_\priv$ and the output $Y_i$ is their message. Note that once the channel $W_i$ is fixed the output distribution of messages is only a function of $X_i$. We now describe the various communication protocols which restrict how the choice of $W_i$s can be performed. 

We restrict ourselves to \emph{simultaneous message passing} (SMP) protocols of communication, \ie{} noninteractive LDP mechanisms, where the $W_i$s are all selected simultaneously. %
Within SMP protocols, we distinguish between the case where a common random seed (public randomness) is available across players and can be used by them to select the $W_i$s, and the case there is no public randomness available and they must choose the $W_i$s independently. In both cases, however, the players are assumed to have access to private randomness, which is needed to implement any privatization mechanism. We describe these two cases in more detail below.

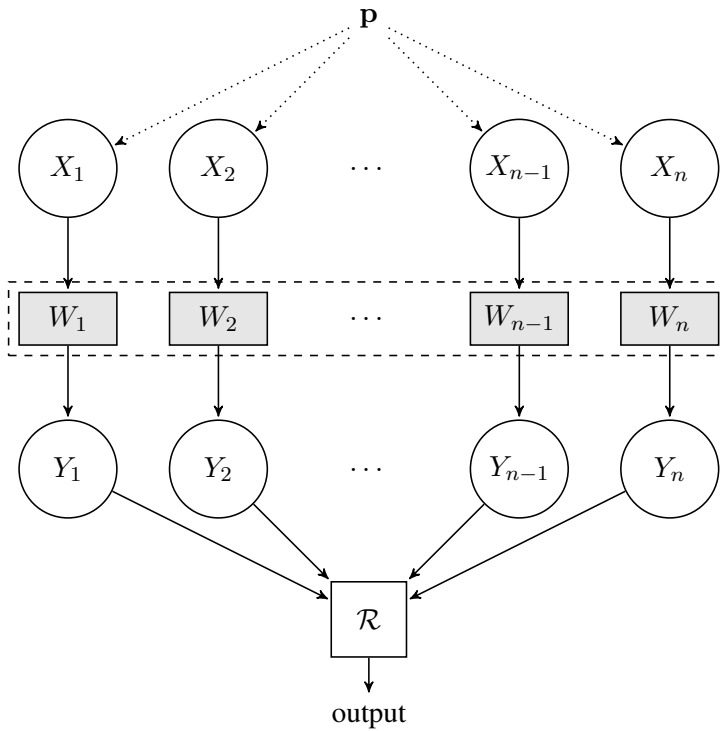
\begin{figure}[ht]\centering
\caption{\label{fig:model}The locally private distributed model, where each $Y_i\in\cY$. In the private-coin setting the channels $W_1,\dots,W_\ns$ are independent, while in the public-coin setting they are jointly randomized.}
\begin{tikzpicture}[->,>=stealth',shorten >=1pt,auto,node distance=20mm, semithick]
\label{tabll}  \node[circle,draw,minimum size=13mm] (A)
  {$X_1$}; \node[circle,draw,minimum size=13mm] (B) [right of=A]
  {$X_2$}; \node (C) [right of=B] {$\dots$}; \node[circle,draw,minimum
  size=13mm] (D) [right of=C] {$X_{\ns-1}$}; \node[circle,draw,minimum
  size=13mm] (E) [right of=D] {$X_\ns$};
  
  \node[rectangle,draw,minimum width=13mm,minimum
  height=7mm,fill=gray!20!white] (WA) [below of=A]
  {$W_1$}; \node[rectangle,draw,minimum width=13mm,minimum
  height=7mm,fill=gray!20!white] (WB) [below of=B] {$W_2$}; \node (WC)
  [below of=C] {$\dots$}; \node[rectangle,draw,minimum
  width=13mm,minimum height=7mm,fill=gray!20!white] (WD) [below of=D]
  {$W_{\ns-1}$}; \node[rectangle,draw,minimum width=13mm,minimum
  height=7mm,fill=gray!20!white] (WE) [below of=E] {$W_\ns$};
  
  \node[draw,dashed,fit=(WA) (WB) (WC) (WD) (WE)] {};
  o \node[circle,draw,minimum size=13mm] (YA) [below of=WA]
  {$Y_1$}; \node[circle,draw,minimum size=13mm] (YB) [below of=WB]
  {$Y_2$}; \node (YC) [below of=WC]
  {$\dots$}; \node[circle,draw,minimum size=13mm] (YD) [below of=WD]
  {$Y_{\ns-1}$}; \node[circle,draw,minimum size=13mm] (YE) [below
  of=WE] {$Y_\ns$};

  \node (P) [above of=C] {$\p$}; \node[rectangle,draw, minimum
  size=10mm] (R) [below of=YC] {$\referee$}; \node (out) [below
  of=R,node distance=13mm] {output};

  \draw[->] (P) edge[dotted] (A)(A) edge (WA)(WA) edge (YA)(YA) edge
  (R); \draw[->] (P) edge[dotted] (B)(B) edge (WB)(WB) edge (YB)(YB)
  edge (R); \draw[->] (P) edge[dotted] (D)(D) edge (WD)(WD) edge
  (YD)(YD) edge (R); \draw[->] (P) edge[dotted] (E)(E) edge (WE)(WE)
  edge (YE)(YE) edge (R); \draw[->] (R) edge (out);
\end{tikzpicture}
\end{figure}

\begin{definition}[Private-coin SMP Protocols]\label{d:private-prot}     
Let $U_1, \dots, U_\ns$ be independent random variables, which are also independent jointly of $(X_1, \dots, X_\ns)$. $U_i$ is the private randomness available to player $i$. %
 A $\priv$-LDP \emph{private-coin} SMP protocol $\pi$ consists of the following two steps: (a)~Player $i$ selects their channel $W_i\in\cW_\priv$ (possibly as a function of $U_i$), (b)~and sends their message $Y_i\in\cY$, which is obtained by passing $X_i$ through $W_i$, to the referee. The referee receives the messages $(Y_1, \dots, Y_\ns)\eqdef \pi(X^\ns)$. We assume that the protocol is decided ahead of time, so that the distribution of the $U_i$ is known to the referee, but not their instantiation.
\end{definition}

Since the random variables $X_i$ and $U_i$ are independent across players, and the message $Y_i$ from player $i$ is a randomized function of $(X_i, U_i)$ the messages $(Y_1,\dots,Y_\ns)$ are all independent across players. 

\begin{definition}[Public-coin SMP  Protocols]
In addition to the private randomness $U_1, \dots, U_\ns$ at the players as above, %
let $V$ be a random variable jointly independent of the the random variables $X_i$ and $U_i$, which denotes the public randomness and is available to all players. A $\priv$-LDP \emph{public-coin} SMP protocol $\pi$ consists of the following two steps: (a)~Player $i$ selects their channel $W_i\in\cW_\priv$ as a function of $V$ (and possibly of $U_i$), and (b)~sends their messages $Y_i\in\cY$,  by passing $X_i$ through $W_i$, to the referee. The referee receives the messages $(Y_1, \dots, Y_\ns)\eqdef \pi(X^\ns, V)$ and the public randomness $V$, but does not have access to the private randomness $(U_1,\ldots, U_\ns)$ of the players.
\end{definition}
 In contrast to private-coin protocols, in a public-coin SMP protocol, the message $Y_i$ from player $i$ is a function of $V$ as well as $(X_i, U_i)$, so the resulting messages $Y_i$ are not independent. They are, however, independent conditioned on the shared randomness $V$. %

We emphasize that private randomness is available even in the public-coin setting and, as previously mentioned, is \emph{required} in order for the protocol to satisfy local privacy (\cf~\cref{def:ldp:channel}).
 This is because the channels must satisfy the LDP condition even when all the information available to the referee, including the public randomness $V$, is fully ``leaked.''
\subsection{Distributed inference protocols}
  
We now provide the formal description of the distributed inference tasks considered
in this work, identity and independence testing.

\paragraph{Identity Testing} Let  $\q\in \distribs{\ab}$ be a known reference distribution. In the $(\ab, \dst, \delta)$-identity testing
problem, we seek to use $\ns$ i.i.d.\ samples from an unknown $\p\in\distribs{\ab}$ to test if $\p$ equals $\q$ or if it is $\dst$-far from $\q$ in total variation distance. A private-coin (resp., public-coin) \emph{$\priv$-LDP protocol for $(\ab,\dst,\delta)$-identity testing} then consists of a private-coin (resp. public-coin) $\priv$-LDP protocol $\pi$ along with a (randomized) mapping $\Tester\colon\cY^\ns\to\{0,1\}$ such that
\begin{align*}
\probaDistrOf{X^\ns\sim \p^\ns}{\Tester(\pi(X^\ns))=1}>1-\delta, &\text{ if
} \p=\q,\\ \probaDistrOf{X^\ns\sim \p^\ns}
             {{\Tester(\pi(X^\ns))=0}}>1-\delta, &\text{ if }
             \totalvardist{\p}{\q}>\dst.
\end{align*} 
Namely, after running the protocol $\pi$ on the independent samples $X^\ns$ held by the players, the referee applies the mapping $\Tester$ to the resulting messages $(Y_1,\ldots, Y_\ns)=\pi(X^\ns)$, which should ``accept'' with high constant probability if the samples come
from the reference distribution $\q$ and ``reject'' with high constant
probability if they come from a distribution significantly far from
$\q$. The special case of identity testing for $\uniformOn{\ab}$ is termed the $(\ab,\dst,\delta)$-\emph{uniformity testing} problem.

The sample complexity of private-coin (resp. public-coin) $\priv$-LDP $(\ab, \dst,\delta)$-identity testing is the minimum $\ns$ for which a $\priv$-LDP protocol for $(\ab,\dst,\delta)$-identity testing with $\ns$ players exists for $\q$. While this quantity can depend on the reference distribution $\q$, it is customary to consider sample complexity over the worst-case
$\q$.\footnote{The sample complexity for a fixed $\q$, without privacy constraints, has been studied
  under the ``instance-optimal'' setting (see~\cite{ValiantV17a, BCG:17}); and under local privacy constraints by~\cite{BB:20}. See also~\cref{sec:identity:testing} for a discussion of the relation between worst-case and instance-optimal settings.}
  
  \paragraph{Independence Testing} In the $(\ab, \dst, \delta)$-independence testing
problem, we seek to use samples from an unknown $\p\in\distribs{[\ab]\times[\ab]}$ (with unknown marginals $\p_1,\p_2\in\distribs{\ab}$) to test if $\p$ equals $\p_1\otimes\p_2$ or if it is $\dst$-far from \emph{every} product distribution in total variation distance. A private-coin (resp., public-coin) \emph{$\priv$-LDP protocol for $(\ab,\dst,\delta)$-independence testing} then consists of a private-coin (resp. public-coin) $\priv$-LDP protocol $\pi$ along with a (randomized) mapping $\Tester\colon\cY^\ns\to\{0,1\}$ such that
\begin{align*}
\probaDistrOf{X^\ns\sim \p^\ns}{\Tester(\pi(X^\ns))=1}>1-\delta, &\text{ if
} \p=\p_1\otimes\p_2,\\ \probaDistrOf{X^\ns\sim \p^\ns}
             {{\Tester(\pi(X^\ns))=0}}>1-\delta, &\text{ if }
             \inf_{\q_1,\q_2\in\distribs{\ab}}\totalvardist{\p}{\q_1\otimes\q_2}>\dst.
\end{align*} 
The sample complexity of private-coin (resp. public-coin) $\priv$-LDP $(\ab, \dst,\delta)$-independence testing is the minimum $\ns$ for which a $\priv$-LDP protocol for $(\ab,\dst,\delta)$-independence testing with $\ns$ players exists for $\q$.

\begin{remark}
We note that the formulation above can be generalized to testing independence over $[\ab_1]\times[\ab_2]$ for arbitrary $\ab_1,\ab_2$, or even over general discrete product spaces $[\ab_1]\times\cdots\times[\ab_d]$. Some of our protocols may generalize to these more general settings, but for simplicity with focus on the simple and arguably fundamental case of independence over $[\ab]\times[\ab]$.
\end{remark}

\section{Locally private identity testing}\label{sec:identity:testing}
We begin by recalling lower bounds from~\cite{AcharyaCT:IT1} which show that for $\priv\in [0,1)$, a private-coin protocol $\priv$-LDP identity testing protocol requires at least $\Omega(\ab^{3/2}/\priv^2\dst^2)$ players and a public-coin protocol requires at least $\Omega(\ab/\priv^2\dst^2)$ players. In this section, we propose both private- and public-coin protocols that attain these bounds, establishing a strict separation between the sample complexity of private- and public-coin protocols. In addition, we give protocols with optimal sample complexity for both settings that require only $1$ bit of communication per player.

\new{We note that our upper bounds for identity are phrased in terms of the domain size $\ab$, or, equivalently, as a worst-case among all possible reference distributions $\q$. However, they immediately imply more refined bounds parameterized by a functional of the reference $\q$ itself (\ie{} ``instance-optimal'' bounds, to follow~\cite{ValiantV17a}) \textit{via} the reduction described in~\cite[Appendix~D]{ACT:19}.}

  \subsection{Private-coin protocols}

We now present private-coin protocols based on \Rappor and
Hadamard Response that are both sample-optimal, with different
communication requirements. 
\subsubsection{A mechanism  based on \Rappor}\label{ssec:identity:rappor} 
We begin by describing the \emph{randomized aggregatable
privacy-preserving ordinal response} (\Rappor) mechanism, which is a
$\priv$-LDP mechanism introduced in~\cite{ErlingssonPK14}.  Its
simplest implementation, $\ab$-\Rappor, maps $\cX=[\ab]$ to 
$\cY=\{0,1\}^\ab$ in two steps.  First, ``one-hot encoding'' is applied
to the input $x\in[\ab]$ to obtain the vector $y'\in\{0,1\}^\ab$ such
that $y'_{j}=\indic{x=j}$ for all $j\in\cX$.  The privatized output
$y \in \cY$ of $\ab$-\Rappor is then a $\ab$-bit vector obtained by
flipping each bit of $y'$ independently with probability
$\frac{1}{e^{\priv/2}+1}$.

Note that if $X$ is drawn from $\p\in\distribs{\ab}$, this leads to
$Y\in\{0,1\}^\ab$ such that the coordinates are (correlated) Bernoulli
random variables, with $Y_j$ distributed as
$\bernoulli{\alprappor\cdot\p(j)+\betrappor}$, $j\in [\ab]$, with
$\alprappor, \betrappor$ defined as
\begin{equation}\label{eq:rappor:parameters}
    \alprappor \eqdef \frac{e^{\priv/2}-1}{e^{\priv/2}+1}
    = \frac{\priv}{4}+o(\priv), \qquad \betrappor \eqdef \frac{1}{e^{\priv/2}+1}
    = \frac{1}{2} + o(\priv).
\end{equation}

Given $\ns$ independent samples from $\p$, let the output of \Rappor
applied to these samples be denoted by $Y_1, \ldots,
Y_\ns\in\{0,1\}^\ab$, where $Y_i = (Y_{i1}, \ldots, Y_{i\ab})$ for
$i\in[\ns]$. The following fact is a simple consequence of the
definition of \Rappor.
\begin{fact}\label{fact:rappor:statistics}
Let $i,j\in[\ns]$, and $x,y\in[\ab]$.
\[
    \probaOf{Y_{ix} =1, Y_{jy}=1} = \begin{cases}
    (\alprappor\p(x)+\betrappor)(\alprappor\p(y)+\betrappor),& \text{
    if } i\neq j\\
    (\alprappor\p(x)+\betrappor)(\alprappor\p(y)+\betrappor)-\alprappor^2\p(x)\p(y),
    & \text{ if } i=j,\, x\neq y\\ \alprappor \p(x) + \betrappor,
    & \text{ if } i=j,\, x= y,\\ \end{cases}
\]
where $\alprappor,\betrappor$ are defined as
in~\eqref{eq:rappor:parameters}. \new{Note that vectors $Y_{i}$ and
$Y_j$ are independent for distinct $i,j\in[\ns]$.}
\end{fact}

We now propose our testing mechanism based on \Rappor, which, in
essence, uses a privatized version of a $\chi^2$-type statistic
of~\cite{ChanDVV14,AcharyaDK15,ValiantV17a}. We note that our
choice of using such a $\chi^2$-type statistic instead of a (perhaps
more natural) ``collision-based'' unbiased estimator for
$\normtwo{\p}^2$ stems from the fact the latter has a high
variance,
leading to a suboptimal sample complexity. For $x\in[\ab]$, let the 
number of occurrences of $x$ among the $\ns$ (privatized) outputs
of \Rappor be 
\begin{equation}\label{eq:rappor:nx}
    N_x \eqdef \sum_{j=1}^\ns \indic{Y_{jx}=1},
\end{equation}
which by the definition of \Rappor follows a
$\binomial{\ns}{\alprappor\p(x)+\betrappor}$ distribution. We consider the following
test statistic $T$:
\begin{equation}\label{eq:rappor:z}
    T \eqdef \sum_{x\in[\ab]} \Paren{ \left(   N_x-(\ns-1)\left(\alprappor\q(x)+\betrappor\right) \right)^2 - N_x + (\ns-1)\left(\alprappor\q(x)+\betrappor\right)^2 }.
\end{equation}
This statistic is motivated from the fact that it constitutes an unbiased
estimator of the the squared $\lp[2]$ distance between $\p$ and $\q$. Using this property
we threshold $T$ to test whether $\p=\q$ or not.
Keeping in mind that $N_x$ is
typically concentrated around its expected value of roughly
$\ns/2$, our new statistic can be seen to take the form
\[
    T \approx \sum_{x\in[\ab]} \Paren{ N_x^2 - \ns N_x }
    + \Theta(\ab\ns^2),
\]
since $\betrappor \approx 1/2$.
In particular,  the subtracted linear
term reduces the fluctuation of the quadratic part, bringing down the variance of the statistic.

\begin{algorithm}[ht]
  \begin{algorithmic}[1]
  \Require Privacy parameter $\priv>0$, distance parameter $\dst\in(0,1)$, $\ns$ players
  \State Set
  \[
      \alprappor \gets \frac{e^{\priv/2}-1}{e^{\priv/2}+1}, \qquad \betrappor \gets \frac{1}{e^{\priv/2}+1}
  \]
  as in~\eqref{eq:rappor:parameters}.
  \State Player $i$ applies ($\priv$-LDP) \Rappor to $X_i$, sends result $Y_i \in \{0,1\}^\ab$ \Comment{Time $O(\ab)$ per user}
  \State Server computes $N_x$ for every $x\in[\ab]$, as defined in~\eqref{eq:rappor:nx} \Comment{Time $O(\ab \ns)$}
  \State Server computes $T$, as defined in~\eqref{eq:rappor:z} \Comment{Time $O(\ab)$}
  \If{$T<\ns(\ns-1)\alprappor^2\dst^2/\ab$}\label{alg:rappor:chisquare:threshold}
    \State \Return \accept
  \Else
    \State \Return \reject
  \EndIf
  \end{algorithmic}
  \caption{Locally Private Identity Testing using \Rappor} \label{alg:rappor:chisquare}
\end{algorithm}

This motivates our testing protocol,~\cref{alg:rappor:chisquare}, and
leads to the main result of this section below.
\begin{theorem}\label{theo:identity:private:rappor}
For every $\ab\geq 1$ and $\priv\in(0,1]$, there exists a private-coin $\priv$-LDP protocol for $(\ab,\dst,\delta)$-identity testing over~$[\ab]$ using \Rappor and $\ns=\bigO{\frac{\ab^{3/2}}{\dst^2\priv^2}\log\frac{1}{\delta}}$ players.
\end{theorem}
\begin{proof}
Each player reports its data using \Rappor, which is a $\priv$-LDP mechanism.
Thus, we only need to analyze the error performance of the proposed test, which we do simply by using Chebyshev's
inequality. Towards that, we evaluate the expected value and the
variance of $T$.

The following evaluation of expected value of statistic $T$ uses a
simple calculation entailing moments of a Binomial random variable:
\begin{lemma}\label{lemma:identity:private:rappor:expect}
  For $T$ defined in~\eqref{eq:rappor:z}, we have
  \[
  \expect{T} = \ns(\ns-1)\alprappor^2\normtwo{\p-\q}^2,
  \] where the expectation is taken over the private coins used by \Rappor and the samples drawn from $\p$. In particular, (i)~if
    $\p=\q$, then $\expect{T}=0$; and (ii)~if
    $\totalvardist{\p}{\q}>\dst$, then $\expect{T}>4\ns(\ns-1)\frac{\alprappor^2\dst^2}{\ab}$.
\end{lemma}
\begin{proof}
    Letting $\lambda_x \eqdef \alprappor\q(x)+\betrappor$, $\mu_x \eqdef \alprappor\p(x)+\betrappor$ for $x\in[\ab]$, and
    using the fact that $N_x$ is Binomial with parameters $\ns$ and $\mu_x$, we
    have
    \begin{align*}
    \expect{T} &= \sum_{x\in[\ab]} \expect{ \left( N_x-(\ns-1)\lambda_x \right)^2 - N_x +  (\ns-1) \lambda_x^2}\\
    & = \sum_{x\in[\ab]} \left( \expect{ N^2_x - N_x}  - 2 (\ns - 1)\lambda_x \expect{N_x} + \left( (\ns-1)^2 + \ns - 1\right)\lambda_x^2 \right)\\
    & = \sum_{x\in[\ab]} \left( \ns(\ns-1)\mu_x^2  - 2 \ns (\ns - 1)\lambda_x \mu_x + \ns(\ns-1)\lambda_x^2 \right)\\
&= \sum_{x\in[\ab]} \ns(\ns-1)\Paren{ \lambda_x-\mu_x }^2\\
&= \ns(\ns-1) \alprappor^2 \sum_{x\in[\ab]}\Paren{ \p(x)-\q(x) }^2.
    \end{align*}
    Claim (i) is immediate; claim (ii) follows upon noting that $\totalvardist{\p}{\q} = \frac{1}{2}\normone{\p-\q} \leq \frac{\sqrt{\ab}}{2}\normtwo{\p-\q}$.
\end{proof}
Turning to the variance, we are able to obtain the following:
\begin{lemma}\label{lemma:identity:private:rappor:variance}
  For $T$ defined in~\eqref{eq:rappor:z}, we have
  \[
  \var[T]  \leq 2\ab\ns^2 + 5\ns^3 \alprappor^2 \normtwo{\p-\q}^2
  \leq 2\ab\ns^2 + 4\ns \expect{T}\,.
  \]
\end{lemma}
\noindent The proof of this lemma is  technical and relies on the analysis of the covariance of the random variables $(N_x)_{x\in[\ab]}$, in view of bounding quantities of the form $\cov(f(N_x), f(N_y))$. We defer the details to~\cref{app:variance:rappor}.

With these two lemmata, we are in a position to conclude the
argument. 

First, consider the case when $\p=\q$. In this case $\expect{T}=0$ and $\var[T]\leq 2\ab\ns^2$
  by~\cref{lemma:identity:private:rappor:expect,lemma:identity:private:rappor:variance}.
  Therefore, by
  Chebyshev's inequality we get
  \[
  \probaOf{ T \geq \ns^2\frac{\alprappor^2\dst^2}{\ab} } \leq \frac{ \ab^2\var[T] }{ \ns^4\alprappor^4\dst^4 } \leq \frac{2\ab^3}{\ns^2\alprappor^4\dst^4},
  \]
  which is at most $1/3$ for $\ns\geq \frac{3\ab^{3/2}}{\alprappor^2\dst^2}$.  
 
Next, when $\totalvardist{\p}{\q}>\dst$, we get
\begin{align*}
\expect{T}&%
 > 4\frac{\ns(\ns-1)}{\ab}\alprappor^2\dst^2,\\
\var[T]&\leq 2\ab\ns^2 + 4\ns\expect{T}.
\end{align*}
Using Chebyshev's inequality yields
\[
 \probaOf{ T
  < \ns^2\frac{\alprappor^2\dst^2}{\ab} } \leq \probaOf{ T
  < \frac{1}{2}\expect{T} } \leq \frac{ 4\var[T] }{ \expect{T}^2
  } 
\leq \frac{\ab^3}{2(\ns-1)^2\alprappor^4\dst^4}
  + \frac{4\ab}{(\ns-1)\alprappor^2\dst^2}, 
  \] 
  which is at most $1/3$ for $\ns\geq \frac{9\ab^{3/2}}{\alprappor^2\dst^2}+1$
  and $\ab\geq 2$. Recalling that $\alprappor = \Theta(\priv)$ concludes the proof of~\cref{theo:identity:private:rappor}, for probability of error $\delta$ set to $1/3$.

Finally, we can reduce this probability of error to an arbitrary $\delta>0$, at the cost of a $O(\log(1/\delta))$ factor in the number of players, using a standard ``amplification'' argument: repeat independently the protocol on $O(\log(1/\delta))$ disjoint sets of players and taking the majority output.
\end{proof}

\subsubsection{A mechanism based on Hadamard Response}\label{ssec:identity:hadamard}
While sample-optimal among private-coin protocols,~\cref{alg:rappor:chisquare} requires each player to communicate $\ab$ bits. We now present a private-coin protocol that is sample-optimal and requires only $1$ bit of communication per player.

Since we seek to send only a $1$-bit message per player, each player can simply indicate if its observation lies in a subset or not. To make this communication LDP, we flip this bit with appropriate probability. In fact, we divide $\ns$ players into $K$ subgroups and associate a subset $C_j\subset[\ab]$, $1\leq j \leq K$, with the $j$th subgroup. Thus, the bits received at the referee can be viewed as $\ns/K$ independent samples from a product-Bernoulli distribution on $\{0,1\}^K$.\footnote{A $K$-dimensional product-Bernoulli distribution is a distribution over $\{0, 1\}^K$, whose coordinates are independently distributed.}

Suppose that  the mean $\mu(\p)$
of the resulting product-Bernoulli distribution satisfies 
 $\normtwo{\mu(\q)-\mu(\p)}>\alpha$ if $\totalvardist{\p}{\q}\geq \dst$.
Then, we can use a test for mean of product-Bernoulli distributions (see, for instance,~\cite[Section~2.1]{CDKS:17} or~\cite[Lemma 4.2]{CKMUZ:19})
to determine if the mean is $\mu(\q)$ or $\alpha$-far from $\mu(\q)$ in $\lp[2]$ distance.

The key question that remains is how large can $\alpha$ be. 
The answer to this question was provided in~\cite{AcharyaS:19},
which introduced the Hadamard Response (HR) mechanism that uses the Hadamard matrix to select $C_j$s that yield a large $\alpha$.

Formally, the HR mechanism can be described as follows. Let
		$K \eqdef 2^{ \clg{\log_2(\ab+1)} }$, 
	which is the smallest power of two larger than $\ab$, and let $H^{(K)}$ be the $K \times K$ Hadamard matrix.
Note that $K\le 2\ab$. 
Let $C_j$ be the location of $1$s in the $j$th column, $i.e.$, $C_j = \setOfSuchThat{i\in[K]}{H^{(K)}_{ij} = 1}$. For any distribution $\p$ over $[k]$ and $C \subset [K]$, let $\p(C)$ be the probability that a sample from $\p$ falls in set $C$. Here we assign zero probability to elements outside $[k]$. 
The key property of the sets $(C_1, \ldots, C_K)$, which was observed in~\cite{AcharyaS:19}, is the following.
		\begin{lemma} \label{lem:parseval}
			For any two distributions $\p, \q$ over $[k]$, 
			\[
				\sum_{j=1}^k (\p(C_j)	 - \q(C_j))^2 = \frac{K}{4}  \normtwo{\p-\q}^2.
			\]
		\end{lemma}
		\begin{proof}
			Let $\p_K, \q_K$ be $K$-dimensional probability vectors obtained by appending zeros to the end of $\p$ and $\q$, respectively, and
                        let $\p(C) := (\p(C_1), \p(C_2), \dots, \p(C_K))$. By the definition of $\p(C_j)$s, we have
			\[
				\p(C) = \frac{1}{2}\Paren{H^{(K)} \p_K + \mathbf{1}_K}, \,\,\, \q(C) = \frac{1}{2}\Paren{H^{(K)} \q_K + \mathbf{1}_K},
			\]
			where $\mathbf{1}_K$ is an all-one vector of dimension $K$. Hence by the fact that $(H^{(K)})^T	H^{(K)} = K\mathbb{I}$, we obtain
			\begin{align*}
				\sum_{j=1}^k (\p(C_j)	 - \q(C_j))^2 & = \normtwo{\p(C) - \q(C)}^2 = \frac{1}{4} (\p_K - \q_K)^T (H^{(K)})^T	H^{(K)}  (\p_K - \q_K) \\
				& = \frac{K}{4}  \normtwo{\p_K-\q_K}^2 = \frac{K}{4}  \normtwo{\p-\q}^2.\qedhere
			\end{align*}
		\end{proof}

In HR, a player observing $X \in [\ab]$ 
assigned a subset $C_j$ sends a random bit $B_{j}$
with distribution given by
		\begin{equation} \label{eqn:response}
\bPr{B_{j}=1|X}=
		  \begin{cases}
		    \frac{e^\priv}{e^\priv+1}, \text{ if } X \in C_j, \\
		    \frac{1}{e^\priv+1}, \text{ otherwise.} 
		  \end{cases}
		\end{equation}
Let $\mu(\p)$ denote the mean of
the product-Bernoulli distribution induced on bits $(B_1, \ldots, B_K)$ (corresponding to any $K$ players assigned sets $(C_1 ,\ldots, C_K)$)
when the observations of players have distribution $\p$, $i.e.$,
		\[
			\mu(\p)_j \eqdef \bE{\p}{\bPr{B_j=1\mid X}}.
		\]
   Following the same computations as in~\cite{AcharyaS:19}, we have that for all $j \in [K]$
		\[
			\mu(\p)_j = \sum_{x \in C_j} \p(x) \frac{e^\priv}{e^\priv+1} +  \sum_{x \notin C_j} \p(x) \frac{1}{e^\priv+1} = \frac{e^\priv - 1}{e^\priv +1} \p(C_j)	+ \frac{1}{e^\priv + 1}.
		\]
Then, by~\cref{lem:parseval},
\begin{align}
\normtwo{\mu(\p)-\mu(\q)}= \frac{\sqrt{K}(e^\priv-1)}{2(e^{\priv}+1)}
\normtwo{\p-\q}\geq \frac{(e^\priv-1)}{2(e^{\priv}+1)}\totalvardist{\p}{\q},
\label{eq:isometry}
\end{align}
where we used the observation that $K \ge \ab$. 

Motivated by this observation,
we obtain~\cref{alg:HR} for LDP identity testing.\footnote{Without loss of generality, we assume $K$ divides $\ns$ (as otherwise	we can ignore the last $(\ns - K \flr{\frac{\ns}{K}})$ players without changing the number of samples by a factor of $2$).}
\begin{algorithm}[ht]
  \begin{algorithmic}[1]
  \Require Privacy parameter $\priv>0$, distance parameter $\dst\in(0,1)$, $\ns$ players
  \State Define $C_j = \setOfSuchThat{i\in[K]}{H^{(K)}_{ij} = 1}$, $j\in [K]$.

\State $\ns$ players are divided into $K$ disjoint subgroups of equal size (using an explicit partition fixed ahead of time).
Players in the $j$th subgroup, $j\in[K]$,  are assigned to the set $C_j$, and they
use~\eqref{eqn:response} to generate their output bits (independent copies of $B_j$).

\State Taking one player from
		each block and viewing the resulting collection of
		messages as a length-$K$ binary vector, the referee gets $\ns/K$
                independent copies of $(B_1, B_2,\ldots, B_K)$
		generated a product-Bernoulli distribution on $\{0,1\}^K$
		with mean vector $\mu(\p)$. 

\State The referee uses these $\ns/K$ samples to test
 whether the mean
		vector $\mu(\p)$ is (i)~a prespecified vector $\mu=\mu(\q)\in\R^K$ or (ii)~at $\lp[2]$ distance at least	$\alpha=\dst/2\in(0,1]$ from $\mu(\q)$. It can use the test from, for instance,~\cite[Section~2.1]{CDKS:17}, which requires $\bigO{\sqrt{K}(\log 1/\delta)/\alpha^2}$ samples to do this. It accepts $\q$ if the mean is $\mu(\q)$, and rejects otherwise.
  \end{algorithmic}
  \caption{Locally Private Identity Testing using Hadamard Response} \label{alg:HR}
\end{algorithm}
The result below summarizes the performance of~\cref{alg:HR}.
\begin{theorem}\label{theo:identity:private:onebit}
For every $\ab\geq 1$ and $\priv\in(0,1]$, there exists a private-coin $\priv$-LDP protocol for $(\ab,\dst,\delta)$-identity testing using one bit of communication per player and $\ns=\bigO{\frac{\ab^{3/2}}{\dst^2\priv^2}\log\frac{1}{\delta}}$ players.
\end{theorem}
\begin{proof}
We have already outline the proof in the discussion above.
It is easy to check that the mechanism in~\eqref{eqn:response} is $\priv$-LDP.
 Further, by~\eqref{eq:isometry}, the test in~\cite[Section~2.1]{CDKS:17} gives the correct outcome with probability of error less than $\delta$ if $\ns/K \gtrsim \sqrt{K}\log(1/\delta)/\alpha^2$ with $\alpha=\dst/2$,
                \ie,
		\[
		\ns = \bigO{\ab^{3/2}\frac{(e^\priv +	1)^2 }{(e^\priv - 1)^2 \dst^2}\log\frac{1}{\delta} } = \bigO{ \frac{\ab^{3/2}}{\dst^2\priv^2}\log\frac{1}{\delta}}\,,
		\]
		suffices as claimed. 		
		\end{proof}
		\begin{remark}
		  \label{rk:privatecoin:l2:guarantee}
		  From the proof of~\cref{theo:identity:private:onebit}, it is clear that the protocol provides a stronger, $\lp[2]$, guarantee: it allows one to distinguish with probability $1-\delta$ between $\normtwo{\p-\q}^2 \leq \frac{\dst^2}{\ab}$ and $\normtwo{\p-\q}^2 \geq \frac{4\dst^2}{\ab}$ with $\ns=\bigO{\frac{\ab^{3/2}}{\dst^2\priv^2}\log\frac{1}{\delta}}$ players (by Cauchy--Schwarz, this implies the total variation testing guarantee). 		  
		  Moreover, the protocol does not require the players to have knowledge of the reference distribution $\q$; it is sufficient that the referee knows it. Both these points are useful, later, for our independence testing results.
		\end{remark}

  \subsection{Public-coin protocols}
    \label{sec:identity-public}
The HR based identity testing protocol generates samples from a product-Bernoulli distribution by assigning different subsets to different subgroups of players. Specifically, we found subsets such that, for any two distributions $\p$ and $\q$, the $\lp[2]$ distance between
the means of the induced approximately $\ab$-dimensional
product distributions is roughly equal to the $\lp[2]$ distance between $\p$ and $\q$.

Interestingly, we can interpret Lemma~\ref{lem:parseval} to get that
for $I$ distributed uniformly over $[K]$, \\$\bEE{(\p(C_I)-\q(C_I))^2}\geq \normtwo{\p-\q}^2/4$. This suggests the possibility of finding a random subset $S$ such that $(\p(S)-\q(S))^2\gtrsim \dst^2/\ab$ if $\totalvardist{\p}{\q}\geq \dst$. Such a set is very handy:
We can simply implement a version of~\cref{alg:HR} with $K=1$ using this set and get a test that works with roughly $\ab/(\rho^2\dst^2$) samples.
This saving in sample-complexity arises from the fact that we were able to retain the same ``per dimension'' $\lp[2]$ distance as that using HR, while using much smaller (only one) dimensional observations. 
But the players need to use public coins to share this set $S$. We formalize
this protocol in this section. 

The first component of our protocol is the following lemma from~\cite{ACT:19}, specialized to a target domain of size $2$. 
\begin{theorem}[{\cite[Theorem~VI.2]{ACT:19}}]
  \label{theo:random:subset}
  Fix any $\ab$-ary distributions $\p,\q$. If $S\subseteq[\ab]$ is a set chosen uniformly at random, we have the following. (i)~if $\p=\q$, then 
  $\p(S)=\q(S)$ with probability one;
  and (ii)~if $\totalvardist{\p}{\q} > \dst$, then
  \[
    \probaDistrOf{S}{ \Paren{\p(S)-\q(S)}^2 >
      \frac{\dst^2}{2\ab} } \geq c\,.
  \]
where $c=1/228$. 
\end{theorem}
\cnote{The constant comes from the proof of Theorem~VI.2. For our choice of u.a.r. set, we have $\probaOf{Y_1\neq Y_2}=1/2$, so we get $\alpha = 9/2048 > 1/228$.}
Thus, indeed, we can find our desired random set $S$.

Next, we present an LDP protocol for testing the bias of coins, our LDP identity testing problem for $\ab=2$. The protocol below
can be viewed as a special case of our protocol in~\cref{ssec:identity:hadamard}; we include this simpler result here for completeness. We have the following.
\begin{lemma}[Locally Private Bias Estimation, Warmup]
      \label{lem:binary-testing}
For every $\priv\in(0,1]$, there exists a private-coin $\priv$-LDP protocol for $(2,\dst,\delta)$-identity testing using one bit of communication per player and $\ns=\bigO{\frac{1}{\dst^2\priv^2}\log\frac{1}{\delta}}$ players. Moreover, the players do not need to know the reference distribution. 
\end{lemma}
\begin{proof}
Assume without loss of generality that the reference distribution is $\q=\bernoulli{q}$. The algorithm uses a simple Randomized Response (RR) scheme~\cite{Warner:65}, where each sample is flipped with probability $1/(e^\priv+1)$. When the input is $\bernoulli{p}$, the output distribution is $\bernoulli{(1+p(e^\priv-1)/(e^\priv+1))}$. Therefore, if $p -q>\dst$, then the bias of the output distribution of applying RR to $\bernoulli{p}$, and $\bernoulli{q}$ differ by $(p-q)(e^\priv-1))/(e^\priv+1)$, which is $\Omega(\dst\priv)$ for $\priv=O(1)$. To distinguish these two Bernoulli distributions with a constant probability, $O(1/(\priv^2\dst^2))$ samples suffice, and the success probability can be boosted to $1-\delta$ by repeating $O(\log(1/\delta)$ times. 
\end{proof}

Motivated by these observations, we propose~\cref{alg:optimal:identity} for public-coin LDP identity testing. 

\begin{algorithm}[ht]
\begin{algorithmic}[1]
\Require Privacy parameter $\priv>0$, distance parameter $\dst\in(0,1)$, $\ns$ players
\State Set
\[
    c \gets \frac{1}{288}\, \qquad \delta_0 \gets \frac{c}{2(1+c)}, \qquad \dst' \gets \frac{\dst}{\sqrt{2\ab}}, \qquad T = \Theta(1), \qquad m \gets \frac{\ns}{T}.
\]
\State Partition the players in $T$ subgroups $G_1,\dots,G_T$ of $m$ players
\For{$t$ from $1$ to $T$} \Comment{In parallel}
  \State Players in $G_t$ generate uniformly at random a common subset $S_t\subseteq [\ab]$.
  \ForAll{$i\in G_t$}
    \State  Player $i$ converts their sample $X_i$ to $X'_i\eqdef \indic{X_i\in S_t}$.
  \EndFor
  \State Players in $G_t$ (and the referee) run the protocol from~\cref{lem:binary-testing} on the samples $(X'_i)_{i\in G_t}$ to test identity of $\p(S_t)$ to $\q(S_t)$, with distance parameter $\dst'$ and failure probability $\delta_0$
\EndFor
\LineComment{At the referee}
\State Let $\tau$ denote the fraction of the $T$ protocols that returned $\accept$
\If{$\tau > 1- (\delta_0 + \frac{c}{4})$}
  \State \Return \accept
\Else
  \State \Return \reject
\EndIf
\end{algorithmic}
\caption{Locally Private Identity Testing}\label{alg:optimal:identity}
\end{algorithm}

We close this section with a characterization of performance of our proposed algorithm.
\begin{theorem}
  \label{theo:identity:public:onebit:restated}
For every $\ab\geq 1$ and $\priv\in(0,1]$, there exists a public-coin $\priv$-LDP protocol for $(\ab,\dst,\delta)$-identity testing using one bit of communication per player and $\ns=\bigO{\frac{\ab}{\dst^2\priv^2}\log\frac{1}{\delta}}$ players.
\end{theorem}
\begin{proof}
  The proof of correctness follows the foregoing outline, which we describe in more detail. Let $c\eqdef 1/288$ be the constant from~\cref{theo:random:subset}, let $\delta_0 \eqdef \frac{c}{2(1+c)}=1/458$, and set $\dst' \eqdef \frac{\dst}{\sqrt{2\ab}}$.
  
Consider the $t$-th test from~\cref{alg:optimal:identity} (where $1\leq t\leq T$), and let $b_t$ be the indicator that the protocol run by players in $G_t$ returned \accept. If $\p=\q$, then by the above we have $\probaOf{ b_t = 1 } \geq 1- \delta_0$ (where the probability is over the choice of the random subset $S_t$, and the randomness of protocol from~\cref{lem:binary-testing}). However, if  $\p$ is $\dst$-far from $\q$, by~\cref{theo:random:subset} it it the case that $\probaOf{ b_t = 1 } \leq (1-c) + c\delta_0 = 1- (\delta_0+\frac{c}{2})$. Therefore, for a sufficiently large constant in the choice of $T = \Theta(1/c^2) = \Theta(1)$, a Chernoff bound argument ensures that we can distinguish between these two cases with probability at least $2/3$.
\end{proof}

    \section{Locally private independence testing}
\label{sec:independence}
In this section, we establish the sample complexity of testing independence of discrete distrbutions. We present private-coin and public-coin protocols for LDP independence testing that require $\bigOmega{\frac{\ab^{3}}{\dst^2\priv^2}}$
and $\bigOmega{\frac{\ab^{2}}{\dst^2\priv^2}}$ players, respectively.
In fact, we show matching lower bounds for these sample complexities
in the final subsection, establishing their optimality among private-coin and public-coin protocols, respectively. The lower bound is a consequence of a general reduction between independence and uniformity testing, which may be of independent interest.

  \subsection{Private-coin protocols}
  \label{sec:independence-private}
To design a private-coin LDP independence testing protocol using $\bigO{\frac{\ab^3}{\priv^2 \dst^2 }}$ players, the first observation we make
is that we can find a product distribution $\widehat{\p}$ that is $\dst/\ab$-close
in
$\lp[2]$ distance from the product distribution $\p_1\times \p_2$ using
$\bigO{\frac{\ab^3}{\priv^2 \dst^2 }}$ players. When the generating distribution $\p$
is not a product distribution, from the separation between our hypothesis, we know that
$\p$ must have $\lp[2]$ distance exceeding $\dst/\ab$ from the product distribution $\widehat{\p}$ we find which is close to $\p_1\times \p_2$.
After this point,
treating $\widehat{\p}$ as the reference, we can use an private-coin LDP identity testing protocol to test if the samples are generated from
a distribution that is close to the 
(product) reference distribution
(in $\lp[2]$ distance) or far from it.  

Formally, we describe the algorithm in Algorithm~\ref{alg:independence:private}, and present its performance in Theorem~\ref{theo:independence:private}.
\begin{algorithm}[ht]
	\begin{algorithmic}[1]
		\Require Privacy parameter $\priv>0$, distance parameter
		$\dst\in(0,1)$, $\ns = \bigO{\frac{\ab^{3}}{\dst^2\priv^2}}$ players 
		\State Partition the players in two groups, $L$ (``learning'') and $T$ (``testing''), each of size
		$\frac{\ns}{2}$.
		\State Players in group $L$ run a $\priv$-LDP \emph{learning}
		protocol to estimate $\p_1\otimes\p_2$ in $\lp[2]$ distance,
		obtaining $\widehat{\p}_1 \otimes \widehat{\p}_2$ such that
		$\normtwo{\widehat{\p}_1 \otimes \widehat{\p}_2 -\p_1 \otimes
			\p_2}^2\leq \frac{\dst^2}{2\ab^2}$ (using the protocol
		of~\cref{lem:product-learn}).
		\State Players in group $T$ run a $\priv$-LDP \emph{identity testing}
		protocol on $\p$, to distinguish between $\normtwo{\p
			-\widehat{\p}_1 \otimes \widehat{\p}_2}^2\leq
		\frac{\dst^2}{2\ab^2}$ and $\normtwo{\p -\widehat{\p}_1 \otimes
			\widehat{\p}_2}^2 \geq \frac{2\dst^2}{\ab^2}$ (using the
		protocol of~\cref{theo:identity:private:onebit}).
	\end{algorithmic}
	\caption{Locally Private Independence Testing (Private-coin)}\label{alg:independence:private}
\end{algorithm}
\begin{theorem}\label{theo:independence:private}
  For every $\ab\geq 1$ and $\priv\in(0,1]$, there exists a private-coin $\priv$-LDP protocol for $(\ab,\dst,\delta)$-independence testing using one bit of communication per player and $\ns=\bigO{\frac{\ab^{3}}{\dst^2\priv^2}\log\frac{1}{\delta}}$ players, where $\dst\in(0,1]$ is the distance parameter.
   \end{theorem}
\begin{proof}
We first note that Algorithm~\ref{alg:independence:private} can be implemented in the SMP setting. Recall that the protocol of~\cref{theo:identity:private:onebit} does
not require the players to know the reference distribution, and
therefore the protocol can be performed in the SMP setting, where players all send their messages simultaneously to the referee. Indeed, in our case, this reference distribution is the
product distribution $\widehat{\p}_1 \otimes \widehat{\p}_2$ computed
from the messages of the players in $L$, so the fact that the players' messages (from the group $T$) do not require knowledge of the reference distribution is crucial to obtain an SMP protocol.

\begin{lemma} \label{lem:product-learn}
Given samples from a distribution $\p$ over $[\ab] \times [\ab]$ with
marginals $\p_1$ and $\p_2$, there exists a private-coin $\priv$-LDP
protocol with $\bigO{\frac{\ab^3}{\priv^2 \dst^2 }}$ players that
outputs distributions $\widehat{\p}_1$, $\widehat{\p}_2$ such that
$\normtwo{\widehat{\p}_1 \otimes \widehat{\p}_2 -\p_1 \otimes
  \p_2}^2\leq \frac{\dst^2}{2\ab^2}$ with probability at least
$5/6$. Moreover, each player sends one bit. 
\end{lemma}
\begin{proof}
	From the known results on LDP distribution estimation~\cite{KairouzBR16,ASZ:18:HR,AcharyaS:19}, with $\bigO{\frac{\ab}{\priv^2 (\dst/\ab)^2 }}=\bigO{\frac{\ab^3}{\priv^2\dst^2 }}$ players one can under $\priv$-LDP output distributions $\widehat{\p}_1$, $\widehat{\p}_2$ such that 
	\[
		\normtwo{\widehat{\p}_1 - \p_1}^2 \le \frac{\dst^2}{8\ab^2}, \qquad \normtwo{\widehat{\p}_2 - \p_2}^2 \le \frac{\dst^2}{8\ab^2}
	\]
	with probability at least $5/6$.	Whenever this guarantee holds, it implies that
    	\begin{align*}
    	\normtwo{\widehat{\p}_1 \otimes \widehat{\p}_2 -\p_1 \otimes \p_2}^2  
    	&\le 2 \cdot \normtwo{\widehat{\p}_1 \otimes \p_2 - \p_1 \otimes \p_2}^2 + 2 \cdot \normtwo{\widehat{\p}_1 \otimes \widehat{\p}_2 - \widehat{\p}_1 \otimes \p_2}^2  \\
    	&\le 2\Paren{\normtwo{\widehat{\p}_1 - \p_1}^2 + \normtwo{\widehat{\p}_2 - \p_2}^2} 
    	\le \frac{\dst^2}{2\ab^2} 
    	\end{align*}
	proving the lemma. The bound on the per-player communication follows from the protocol of~\cite{AcharyaS:19}.
\end{proof}

Using the protocol from~\cref{lem:product-learn}, we get the following with probability $5/6$. If $\p$ is a product distribution with marginals $\p_1$ and $\p_2$, then 
\[
	\normtwo{\widehat{\p}_1 \otimes \widehat{\p}_2 - \p}^2 = \normtwo{\widehat{\p}_1 \otimes \widehat{\p}_2 - \p_1 \otimes \p_2}^2 \le \frac{\dst^2}{2\ab^2}.
\]
If however $\p$ is $\dst$-far from being a product distribution, then, by the Cauchy--Schwarz inequality,
\[
	\normtwo{\p - \widehat{\p}_1 \otimes \widehat{\p}_2}^2 \ge \frac{4}{\ab^2}\normone{\p - \widehat{\p}_1 \otimes \widehat{\p}_2}^2 > 4\frac{\dst^2}{\ab^2}.
\]
We can use the protocol from~\cref{theo:identity:private:onebit}
(specifically, recalling~\cref{rk:privatecoin:l2:guarantee}) to distinguish the two cases with $\bigO{\frac{(\ab^2)^{3/2}}{\priv^2\dst^2 }}=\bigO{\frac{\ab^3}{\priv^2\dst^2 }}$ players, and probability of success $5/6$. By a union bound over the two protocols used, the overall tester is successful with probability at least $2/3$. Amplifying the probability of success to $1-\delta$ by running the protocol in parallel on $O(\log(1/\delta))$ disjoint sets of players and taking the majority output yields the result.
\end{proof}

  \subsection{Public-coin protocols}
  \label {sec:idependence-public}
We now present our public-coin protocol for LDP independence testing.
Our approach is similar to the one we followed for our public-coin
LDP identity testing protocol: namely, we first use public coins
to ``embed'' the problem in a smaller domain of size $\ab=2$, and then apply an LDP independence test for $\ab=2$. 
For this strategy to work, we first need a
result guaranteeing that randomly
hashing the domain $[\ab]\times[\ab]$ to $\{0,1\}\times\{0,1\}$
{while respecting the product structure} preserves distances. This is what we provide next, establishing an analogue
of~\cref{theo:random:subset} tailored to the product space setting.
\begin{theorem}
  \label{theo:random:subset:product}
  Fix any distribution $\p$ over $[\ab]\times[\ab]$ with marginals
  $\p_1,\p_2$. If $S_1,S_2\subseteq[\ab]$ are two sets chosen
 independently and uniformly at random, we have the following. (i)~if
  $\p=\p_1\otimes\p_2$, then $\p(S_1\times S_2)=\p_1(S_1)\p_2(S_2)$
  with probability one; and (ii)~if
  $\totalvardist{\p}{\p_1\otimes\p_2} > \dst$, then
  \[
    \probaDistrOf{S_1,S_2}{ \Paren{\p(S_1\times
        S_2)-\p_1(S_1)\p_2(S_2)}^2 > \frac{\dst^2}{8\ab} } \geq c\,.
  \] for some absolute constant $c>0$. (Moreover, one can take $c=1/4096$.)
\end{theorem}
We emphasize that Theorem~\ref{theo:random:subset:product} is not a direct consequence
of~\cref{theo:random:subset}, due to the product structure of the
random subset $S_1\times S_2$ (while the previous theorem would apply
to a random subset $S\subseteq [\ab]\times[\ab]$). And indeed,
proving~\cref{theo:random:subset:product} requires the following
hashing lemma, proven in a fashion similar
to~\cite[Theorem~A.6]{ACT:19}:
\begin{theorem}[Joint Probability Perturbation Hashing]
  \label{theorem:random:product:subsets}
Consider a matrix $\delta\in\R^{\ab\times \ab}$ such that, for every
$i_0,j_0\in[\ab]$, $\sum_{j\in [\ab]}\delta_{i_0,j}=\sum_{i\in
  [\ab]}\delta_{i,j_0} = 0$. Let random variables $X=(X_1, \dots,
X_\ab)$ and $Y=(Y_1, \dots, Y_\ab)$ be independent and uniformly
distributed over $\ab$-length binary sequences.  Define $Z =
\sum_{(i,j)\in [\ab]\times[\ab]}\delta_{ij}X_iY_j$.  Then, for every
$\alpha\in(0,1/16)$, there exists a constant $c_\alpha>0$ such that
\[
\probaOf{Z^2\geq \alpha\norm{\delta}_F^2}\geq c_\alpha.
\]
\end{theorem}
The proof of this theorem is deferred
to~\cref{app:concentration:bivariate}. We now show how this
implies~\cref{theo:random:subset:product}:
\begin{proofof}{\cref{theo:random:subset:product}}
Let $\p$ be as in the statement. Item (i) is from the
definition. We just focus on proving item (ii). Define
$\delta\in\R^{\ab\times\ab}$ by $\delta_{ij} = \p(i,j)-\p_1(i)\p_2(j)$
for $i,j\in[\ab]$. Since $\p$ has marginals $\p_1,\p_2$, $\delta$
satisfies the assumptions of~\cref{theorem:random:product:subsets}, we
can apply the theorem, observing that if $X$ (resp. $Y$) is the
indicator vector of the set $S_1$ (resp. $S_2$) then
\[
    Z = \sum_{(i,j)\in [\ab]\times[\ab]}(\p(i,j)-\p_1(i)\p_2(j))X_iY_j
    = \p(S_1\times S_2)-\p_1(S_1)\p_2(S_2)\,,
\]
and that $ \norm{\delta}_F^2 = \normtwo{\p-\p_1\otimes\p_2}^2 \geq
\frac{4\dst^2}{\ab^2} $ (the inequality being Cauchy--Schwarz). Taking 
$\alpha=1/32$ yields the result.
\end{proofof}

It only remains to describe an LDP independence testing protocol for $\ab=2$.
Note that
while we can set $\ab=2$ in the protocol of~\cref{theo:independence:private}, it leads to a complicated protocol. We instead provide a simple test for $\ab=2$.        
\begin{lemma}[Locally Private Bias Estimation]
      \label{lem:binary-testing:independence}
Let $\priv\in(0,1]$. There exists a private-coin $\priv$-LDP protocol
  for $(2,2,\dst,\delta)$-independence testing using one bit of
  communication per player and
  $\ns=\bigO{\frac{1}{\dst^2\priv^2}\log\frac{1}{\delta}}$ players.
\end{lemma}
\begin{proof}
Consider a distribution $\p$ over $\{0,1\} \times \{0,1\}$ with
marginals $\p_1$ and $\p_2$.  We use the fact that
\[
\abs{ \p(0,0)-\p_1(0)\p_2(0)|=|\p(x,y)-\p_1(x)\p_2(y) }, \quad x,y\in
\{0,1\},
\]
which holds since
\[
\abs{ \p(0,1)-\p_1(0)\p_2(1)}= \abs{(\p_1(0) - \p(0,0))-\p_1(0)(1 - \p_2(0))} = \abs{\p_1(0) \p_2(0) - \p(0,0)}.
\]
Thus, if $\p$ is $\dst$-far in total variation distance from every
product distribution, it must hold that
$\totalvardist{\p}{\p_1\otimes\p_2}\geq \dst$, which in view of the
equation above yields $|\p(0,0)-\p_1(0)\p_2(0)|\geq \dst/2$.  Using
this observation, we can test for independence using
$O(1/(\priv^2\dst^2))$ samples as follows.

The $\ns$ players are partitioned in $3$ sets $A,B,C$ of size
$\ns/3$. Since, for any symbol $(x,y)$, $\p(x,y)$ (resp. $\p_1(x)$,
$\p_2(y)$) can be estimated up to accuracy $\dst$ by converting the
observation $(X,Y)$ to the binary observation $\indic{(X,Y)=(x,y)}$
(resp. $\indic{X=x}$, $\indic{Y=y}$) and proceeding as in
\cref{lem:binary-testing}, we can estimate $\p(0,0)$, $\p_1(0)$, and
$\p_2(0)$ up to an additive accuracy $\dst/16$ by assigning
$|A|=|B|=|C|=O(1/(\priv^2\dst^2)\log(1/\delta))$ players for each of
them, so that the three estimates are simultaneously accurate with
probability at least $1-\delta$. Denote these estimates by
$\tilde{\p}(0,0)$, $\tilde{\p}_1(0)$, and $\tilde{\p}_2(0)$,
respectively. When $\p(0,0)=\p_1(0)\p_2(0)$,
\[
|\tilde{\p}(0,0)-\tilde{\p}_1(0)\tilde{\p}_2(0)|\leq
|\tilde{\p}(0,0)-\p(0,0)|+ \tilde{\p}_1(0)-\p_1(0)|+
|\tilde{\p}_2(0)-\p_2(0)| \leq \frac 3 {16} \dst.
\]
On the other hand, when $|\p(0,0)-\p_1(0)\p_2(0)|\geq \dst/2$, we have
\[
|\tilde{\p}(0,0)-\tilde{\p}_1(0)\tilde{\p}_2(0)| \geq
|\p(0,0)-\p_1(0)\p_2(0)|- |\tilde{\p}(0,0)-\p(0,0)|-
|\tilde{\p}_1(0)-\p_1(0)|- |\tilde{\p}_2(0)-\p_2(0)| \geq
\frac{5}{16}\dst.
\]
Thus, it is sufficient for the referee to form the estimates
$\tilde{\p}(0,0)$, $\tilde{\p}_1(0)$, and $\tilde{\p}_2(0)$ and
compare $\abs{ \tilde{\p}(0,0)-\tilde{\p}_1(0)\tilde{\p}_2(0) }$ to
the threshold $\dst/4$.
\end{proof}

We summarize the overall algorithm and its performance below.
\begin{algorithm}[ht]
\begin{algorithmic}[1]
\Require Privacy parameter $\priv>0$, distance parameter
$\dst\in(0,1)$, $\ns$ players \State Set
\[
    c \gets \frac{1}{4096}\, \qquad \delta_0 \gets
    \frac{c}{2(1+c)}, \qquad \dst' \gets \frac{\dst}{\sqrt{8\ab}},
    \qquad T = \Theta(1), \qquad m \gets \frac{\ns}{3T}.
\]
\State Partition the players in $3T$ groups
$B_{1,1},B_{1,2},B_{1,3},B_{2,1},B_{2,2},B_{2,3},\dots,B_{T,1},B_{T,2},B_{T,3}$
of $m$ players \For{$t$ from $1$ to $T$} \Comment{In parallel} \State
Players in $B_{t,1}\cup B_{t,2}\cup B_{t,3}$ generate uniformly at
random two common subsets $S_{1,t},S_{2,t}\subseteq [\ab]$.
\ForAll{$i\in B_{t,1}$} \State Player $i$ converts their sample
$(X_i,Y_i)$ to $X'_i\eqdef \indic{(X_i,Y_i)\in S_{t,1}\times
  S_{t,2}}$.  \EndFor \ForAll{$i\in B_{t,2}$} \State Player $i$
converts their sample $(X_i,Y_i)$ to $X'_i\eqdef \indic{X_i\in
  S_{t,1}}$.  \EndFor \ForAll{$i\in B_{t,3}$} \State Player $i$
converts their sample $(X_i,Y_i)$ to $X'_i\eqdef \indic{Y_i\in
  S_{t,2}}$.  \EndFor \State Players in $B_{t,1}\cup B_{t,2}\cup
B_{t,3}$ (and the referee) run the protocol
from~\cref{lem:binary-testing} on the samples $(X'_i)_{i\in
  B_{t,1}\cup B_{t,2}\cup B_{t,3}}$ to test identity of
$\p(S_{t,1}\times S_{t,2})$ to $\p_1(S_{t,1})\p_1(S_{t,2})$, with
distance parameter $\dst'$ and failure probability $\delta_0$ \EndFor
\LineComment{At the referee} \State Let $\tau$ denote the fraction of
the $T$ protocols that returned $\accept$ \If{$\tau > 1- (\delta_0 + \frac{c}{4})$}
  \State \Return \accept
\Else
  \State \Return \reject
\EndIf
\end{algorithmic}
\caption{Locally Private Independence Testing (Public-coin)}\label{alg:optimal:independence}
\end{algorithm}

\begin{theorem}
  \label{theo:independence:public:onebit:restated}
For every $\ab\geq 1$ and $\priv\in(0,1]$, there exists a public-coin $\priv$-LDP protocol for $(\ab,\dst,\delta)$-independence testing using one bit of communication per player and $\ns=\bigO{\frac{\ab^2}{\dst^2\priv^2}\log\frac{1}{\delta}}$ players.
\end{theorem}
\begin{proof}
  The proof proceeds as follows:
   Using the public randomness, the players select two uniformly random subsets $S_1,S_2\subseteq [\ab]$, and from their samples allow the referee to estimate the quantities $\p(S_1\times S_2)$, $\p_1(S_1)$, and $\p_2(S_2)$. By~\cref{theo:random:subset:product}, this in turn is enough to detect (with constant probability over the choice of $S_1,S_2$) if $\p$ is far from $\p_1\otimes\p_2$; it then suffices to repeat this in parallel on disjoint groups of players in order to amplify the probability of success.

To wit,  the proof of correctness follows the foregoing outline, which we describe in more detail. Let $c\eqdef 1/4096$ be the constant from~\cref{theo:random:subset:product}, let $\delta_0 \eqdef \frac{c}{2(1+c)}$, and set $\dst' \eqdef \frac{\dst}{\sqrt{8\ab}}$.
  
Consider the $t$-th test from~\cref{alg:optimal:independence} (where $1\leq t\leq T$), and let $b_t$ be the indicator that the protocol run by players in $B_t$ returned \accept. If $\p=\p_1\otimes\p_2$, then by the above we have $\probaOf{ b_t = 1 } \geq 1- \delta_0$ (where the probability is over the choice of the random subsets $S_{t,1}$ and $S_{t,2}$, and the randomness of protocol from~\cref{lem:binary-testing:independence}). However, if  $\p$ is $\dst$-far from $\p_1\otimes\p_2$, by~\cref{theo:random:subset:product} it it the case that $\probaOf{ b_t = 1 } \leq (1-c) + c\delta_0 = 1- (\delta_0+\frac{c}{2})$. Therefore, for a sufficiently large constant in the choice of $T = \Theta(1/c^2) = \Theta(1)$, a Chernoff bound argument ensures that we can distinguish between these two cases with probability at least $2/3$.
\end{proof}

  \subsection{Lower bounds}
\label{sec:lower-independence}
The following theorem proves the tightness of our upper bounds for independence testing. 
\begin{theorem}
  \label{corollary:independence:uniformity:reduction}
For every $\ab\geq 1$ and $\priv\in(0,1]$, every private-coin (resp., public-coin) $\priv$-LDP protocol for $(\ab,\dst,1/12)$-independence testing must have $\bigOmega{\frac{\ab^{3}}{\dst^2\priv^2}}$ players (resp., $\bigOmega{\frac{\ab^{2}}{\dst^2\priv^2}}$ players).
\end{theorem}
\begin{proof}
We show the following reduction, which implies our bounds for independence testing.
If there exists a private-coin (resp., public-coin) $\priv$-LDP protocol for $(\ab,\dst,1/12)$-independence testing with $\ns$ players, then there also exists a private-coin (resp., public-coin) $\priv$-LDP protocol for distinguishing the ``Paninski construction'' over $[\ab^2]$ with $\ns$ players. Recall that for even integer $\ab$ and a distance parameter $\gamma\in[0,1/2]$, the Paninski construction is a family of $2^{\ab^2/2}$ distributions  $\{\p_z\}_{z\in\bool^{\ab^2/2}}$ over $[\ab^2]$, where for $z\in\bool^{\ab^2/2}$ we have
\begin{equation}
	\label{eq:paninski}
	\p_z(x) =
	\begin{cases}
		\frac{1-2\gamma z_i}{\ab^2}, &  x=2i-1\\
		\frac{1+2\gamma z_i}{\ab^2}, & x=2i
	\end{cases},\qquad x\in[\ab^2]\,.
\end{equation}
Note that every $\p_z$ is then at total variation distance exactly $\gamma$ from $\uniformOn{\ab^2}$. 
From the lower bounds on uniformity testing already established in~\cite{AcharyaCT:IT1} (listed in \cref{table:results}), we then obtain the lower bounds on independence testing.

We first state the following useful fact (see \eg{},~\cite{BFFKRW:01}) which states that if a distribution is close to a product distribution, then it must be close to the product of its own marginals.  
\begin{fact}
  \label{fact:independence:marginals}
  Let $\p, \q\in\distribs{\Omega\times\Omega}$ with $\q$, a product distribution. If $\totalvardist{\p}{\q}\le \dst$ then
  $\totalvardist{\p}{\p_1\otimes\p_2} \leq 3\dst$.
\end{fact}
Let $\ab=2\ell$. For $z\in\bool^{2\ell^2}$, let $(\p_z)_z$ be the collection of distributions over $[4\ell^2]=[\ab^2]$ given in~\eqref{eq:paninski}, each at a distance $\gamma\eqdef3\dst$ from the uniform distribution. We construct a mapping $\Phi\colon\distribs{[\ab^2]}\to \distribs{[2\ab]\times[2\ab]}$ such that: 
\begin{enumerate}
\item\label{reduction:item:i} Both marginals of $\Phi(\p_z)$ are $\uniformOn{[2\ab]}$ for all $z$;
\item\label{reduction:item:ii} $\totalvardist{\Phi(\p_z)}{\uniformOn{[2\ab]\times[2\ab]}} = \totalvardist{\p_z}{\uniformOn{\ab^2}}$, and $\Phi(\uniformOn{\ab^2})=\uniformOn{[2\ab]\times[2\ab]}$;
\item\label{reduction:item:iii} There exists a mapping from $[\ab^2] \to [2\ab]\times[2\ab]$ that converts a sample from $\p_z$ into a sample from $\Phi(\p_z)$, and a sample from $\uniformOn{[\ab^2]}$ into a sample from $\uniformOn{[2\ab]\times[2\ab]}$.
\end{enumerate}

By~\cref{fact:independence:marginals}, for any product distribution $\q$ over $[2\ab]\times[2\ab]$,
\[
\totalvardist{\Phi(\p_z)}{\q}\ge \totalvardist{\Phi(\p_z)}{\uniformOn{[2\ab]\times[2\ab]}}/3=\totalvardist{\p_z}{\uniformOn{\ab^2}}=\dst,
\]
and the distribution $\Phi(\p_z)$ is at least $\dst$-far from any product distribution. Now, by~\cref{reduction:item:iii}, if we obtain $\ns$ samples from $\Phi(\p_z)$ for a uniformly chosen $z$, we can convert them to $\ns$ samples from $\p_z$. Therefore, any algorithm for testing independence can be used to test uniformity for the Paninski class of distributions over $[\ab^2]$, for which the lower bounds were established in~\cite{AcharyaCT:IT1}. This proves~\cref{corollary:independence:uniformity:reduction}, assuming the mapping $\Phi$.\medskip

We now describe the function $\Phi$ satisfying the three conditions. %
To each $i\in[2\ell^2]$, we associate a collection $C_i=\{a_{i,j}, b_{i,j}\}_{1\leq j\leq 4} \subset [2\ab] \times [2\ab]$ of $8$ elements and arrange them in a ``block'' $B_i$ as
\[
    B_i \eqdef \begin{bmatrix}
      a_{i,1} & b_{i,1} & a_{i,2} & b_{i,2} \\
      b_{i,3} & a_{i,3} & b_{i,4} & a_{i,4}
    \end{bmatrix},
\]
The we can see the set of $(2\ab)^2=8\cdot2\ell^2$ elements $C\eqdef \bigcup_{i=1}^{2\ell^2} C_i$ 
as a $2\ab$-by-$2\ab$ matrix $B$, comprised of the $2\ell^2$ blocks as follows:
\[
    B \eqdef \begin{bmatrix}
      B_1 & B_2 & \dots & B_\ell \\
      B_{\ell+1} & B_{\ell+2} & \dots & B_{2\ell} \\
      \vdots & \vdots & \ddots & \vdots \\
      B_{(2\ell-1)+1} & B_{(2\ell-1)+2} & \dots & B_{2\ell^2}
    \end{bmatrix}.
\]
This matrix $B$ enables us to see the target domain $[2\ab]\times[2\ab]$ as
this $2\ell$-by-$\ell$ grid of $2$-by-$4$ blocks of elements.
Explicitly, this correspondence is given by the indices
\begin{align*}
	a_{i,1} &=\Paren{2 r_i +1 , 4 c_i + 1},\,\,
b_{i,1} =\Paren{2 r_i +1 , 4 c_i + 2},
\\
a_{i,2} &=\Paren{2 r_i +1 , 4 c_i + 3},\,\,
b_{i,2} =\Paren{2 r_i +1 , 4 c_i + 4},
\\
a_{i,3} &=\Paren{2 r_i +1 , 4 c_i + 2},\,\,
b_{i,3} =\Paren{2 r_i +2 , 4 c_i + 1},
\\
a_{i,4} &=\Paren{2 r_i +1 , 4 c_i + 4},\,\,
b_{i, 4} =\Paren{2 r_i +1 , 4 c_i + 3},
\end{align*}
where $r_i = \flr{i/\ell}$ and $c_i = i \bmod \ell$, for $1\leq i \leq 2\ell^2$.
This enables us to define our mapping $\Phi\colon\distribs{[\ab^2]}\to \distribs{[2\ab]\times[2\ab]}$: given a distribution $\p$ over $[\ab^2]=[4\ell^2]$, let, for all $i \in [2\ell^2]$,
\begin{align*}
    \Phi(\p)(a_{i,1})&=\Phi(\p)(a_{i,2})=\Phi(\p)(a_{i,3})=\Phi(\p)(a_{i,4})=\frac{1}{4}\p(2i-1), \\ 
    \Phi(\p)(b_{i,1})&=\Phi(\p)(b_{i,2})=\Phi(\p)(b_{i,3})=\Phi(\p)(b_{i,4})=\frac{1}{4}\p(2i).
\end{align*}
Heuristically,
for each $1\leq i \leq 2\ell^2=\ab^2/2$,
recalling the layout of the block $B_i$,
the mapping $\Phi$ ``distributes'' the probability masses $\p(2i-1)$ and $\p(2i)$ on $8$ elements of $C_i$ as follows:
\[
    \frac{1}{4}\begin{bmatrix}
      \p(2i-1) & \p(2i) & \p(2i-1) & \p(2i) \\
      \p(2i) & \p(2i-1) & \p(2i) & \p(2i-1)
    \end{bmatrix}.
\]
This implies~\cref{reduction:item:ii}, since $\lp[1]$ distance (and thus total variation) is preserved by this transformation. It also establishes~\cref{reduction:item:iii}, as upon seeing a sample $x$ from $\p$, one can generate a sample from $\Phi(\p)$ by returning uniformly at random one of the four corresponding elements from the block $B_{\clg{x/2}}$. Thus, it only remains to show~\cref{reduction:item:i}. This in turn comes from the fact that for every $i\in[2\ell^2]$, by construction, the probabilities under $\Phi(\p)$ of each row (resp., column) of block $B_i$ sum to
$\frac{1}{2}(\p(2i-1)+\p(2i))$ (resp., $\frac{1}{4}(\p(2i-1)+\p(2i))$), which are $1/\ab^2$ and $1/2\ab^2$ respectively for $\p_z$ and $\uniformOn{[\ab^2]}$, independent of $i$.
\end{proof}

\appendices

\appendix
\section*{Proof of~\cref{lemma:identity:private:rappor:variance}}\label{app:variance:rappor}
In this section, we provide the proof of the variance bound for the \Rappor-based statistic of~\cref{ssec:identity:rappor}.
\begin{lemma}[{\cref{lemma:identity:private:rappor:variance}, restated}]
For $T$ defined as in~\cref{eq:rappor:z}, we have 
  \[
  \var[T]  \leq 2\ab\ns^2 + 5\ns^3 \alprappor^2 \normtwo{\p-\q}^2
  \leq 2\ab\ns^2 + 4\ns \expect{T}\,.
  \]
\end{lemma}
\begin{proof}
    We let $\lambda_x \eqdef \alprappor\q(x)+\betrappor$ and $\mu_x\eqdef \frac{1}{\ns}\expect{N_x} = \alprappor\p(x)+\betrappor$ for $x\in[\ab]$.
    Dropping the constant terms from $T$, we define $T'$ such that $\var[T'] = \var[T]$ as
    \[
          T' \eqdef \sum_{x\in[\ab]} \Paren{ N_x^2 - (2(\ns-1)\lambda_x+1) N_x } = \sum_{x\in[\ab]} g(N_x,\lambda_x),
    \]
    where $g\colon[0,\infty)\times[0,1]\to\R$ is given by $g(t,\lambda) = t^2 - (2(\ns-1)\lambda+1) t $. 
The key difficulty in the analysis arises from the fact that the multiplicities of the $N_x$ terms that arise from \Rappor are correlated random variables. 
Because $g$ is not monotone in its first input, the cross covariance terms may be positive even though the $N_x$ terms are negatively associated.
 As a result, we fully expand out the variance and analyze the terms separately. 
 Recall that
\begin{equation}
  \label{eq:identity:rappor:variance:all}
    \var[T'] = \sum_{x\in[\ab]} \var[ g(N_x,\lambda_x) ] + 2\sum_{x<y} \cov( g(N_x,\lambda_x) , g(N_y,\lambda_y) ).
\end{equation}
We first analyze the sum of variances. A direct computation gives that, for every $x\in[\ab]$,
    \begin{align*}
       \var[ g(N_x,\lambda_x) ] 
          &= 2 \ns(\ns-1) \mu_x(1-\mu_x) \Paren{ \mu_x(1-\mu_x) + 2(\ns-1)(\lambda_x-\mu_x)^2}\\
          &\leq \frac{1}{8}\ns^2 + \alprappor^2\ns(\ns-1)^2\left(\p(x)-\q(x)\right)^2,
    \end{align*}
where the inequality holds since $\mu_x\in[0,1]$ so $\mu_x (1-\mu_x) \le 1/4$. It
follows that
    \begin{align}\label{eq:identity:rappor:variance:var}
        \sum_{x\in[\ab]}\var[ g(N_x, \lambda_x) ]  &\leq \frac{1}{8}\ns^2\ab + \alprappor^2\ns(\ns-1)^2\normtwo{\p-\q}^2\,.
    \end{align}
    
We now turn to the sum of the covariance terms. Fix any $x < y$ in $[\ab]$. By expanding the corresponding covariance term, we get
\begin{align}
      \cov( g(N_x,\lambda_x) , g(N_y,\lambda_y) ) 
      &= \expect{ g(N_x,\lambda_x) g(N_y,\lambda_y) } - \expect{ g(N_x,\lambda_x) }\expect{ g(N_y,\lambda_y) } \notag\\
      &=\expect{ N_x^2 N_y^2} 
      - (2(\ns-1)\lambda_y+1)\expect{N_x^2N_y}
      - (2(\ns-1)\lambda_x+1)\expect{N_xN_y^2} \notag\\
      &\qquad+ (2(\ns-1)\lambda_x+1)(2(\ns-1)\lambda_x+1)\expect{N_xN_y} \notag\\
      &\qquad- \ns^2(\ns-1)^2\mu_x\mu_y(\mu_x-2\lambda_x)(\mu_y-2\lambda_y) \label{eq:identity:rappor:cov:0}
\end{align}
since $\expect{N_x^2-(2(\ns-1)\lambda_x+1)N_x} = \ns(\ns-1)\mu_x(\mu_x-2\lambda_x)$. We then proceed by evaluating the expressions for $\expect{N_xN_y}$,
    $\expect{N_x^2N_y}$, $\expect{N_xN_y^2}$, and
    $\expect{N_x^2N_y^2}$ separately.
    
  First, by~\cref{fact:rappor:statistics}, we have that
  \begin{align}
      \expect{N_xN_y} &= \sum_{1\leq i,j\leq \ns}  \probaOf{Y_{ix}=1, Y_{jy}=1}  
       = \sum_{i=1}^\ns ( \mu_x\mu_y - \alprappor^2\p(x)\p(y)) + 2\sum_{i <j} \mu_x\mu_y   \nonumber\\
       &= \ns^2\mu_x\mu_y - \ns\alprappor^2\p(x)\p(y)
       = \ns^2\mu_x\mu_y - \ns(\mu_x-\betrappor)(\mu_y-\betrappor)
       \,.\label{eq:identity:rappor:cov:11}
  \end{align}
  Second, for $\expect{N_x^2N_y}$, we get
    \begin{align*}
        \expect{N_x^2 N_y}
        &= \sum_{1\leq i,j,\ell\leq \ns} \probaOf{ Y_{ix} = 1, Y_{jx} = 1, Y_{\ell y} = 1 } \\
        &= \ns\probaOf{ Y_{ix} = 1, Y_{i y} = 1 } + 6\binom{\ns}{3} \mu_x^2 \mu_y + 2\binom{\ns}{2}\left( \mu_x\mu_y+2\mu_x(\mu_x\mu_y-\alprappor^2\p(x)\p(y)) \right) \\
        &= \ns \mu_x\mu_y - \ns\alprappor^2\p(x)\p(y) + 6\binom{\ns}{3}\mu_x^2 \mu_y  
        + \ns(\ns-1) \mu_x\mu_y+ 4\binom{\ns}{2}\mu_x^2\mu_y - 4\binom{\ns}{2}\alprappor^2\mu_x\p(x)\p(y),
    \end{align*}
which, gathering the terms, yields
    \begin{equation}\label{eq:identity:rappor:cov:21}
        \expect{N_x^2 N_y} = \ns^2 \mu_x\mu_y - (2(\ns-1)\mu_x+1)\ns(\mu_x-\betrappor)(\mu_y-\betrappor) + \ns^2(\ns-1)\mu_x^2 \mu_y.
    \end{equation}
    The term $\expect{N_x N_y^2}$ term follows similarly.
        
  Finally, for $\expect{N_x^2N_y^2}$, note that
    \begin{align*}
        \expect{N_x^2 N_y^2} 
        &= \sum_{1\leq i,j,i',j'\leq \ns} \probaOf{ Y_{ix} = 1, Y_{jx} = 1, Y_{i' y} = 1, Y_{j' y} = 1 } \\
        &= \ns\Paren{\mu_x\mu_y-\alprappor^2\p(x)\p(y)} 
        + \binom{\ns}{2}\Paren{2\mu_x\mu_y + 4\mu_x\Paren{\mu_x\mu_y-\alprappor^2\p(x)\p(y)} \right.\\&\qquad\left.+ 4\mu_y\Paren{\mu_x\mu_y-\alprappor^2\p(x)\p(y)} + 4\Paren{\mu_x\mu_y-\alprappor^2\p(x)\p(y)}^2 } \\
        &\qquad+ \binom{\ns}{3}\Paren{ 6\mu_x^2\mu_y+6\mu_x\mu_y^2 + 24\mu_x\mu_y\Paren{\mu_x\mu_y-\alprappor^2\p(x)\p(y)}  } \\
        &\qquad + 24\binom{\ns}{4} \mu_x^2\mu_y^2
    \end{align*}
    where the second equality follows from counting the different
    possibilities for the values taken by $i,i',j,j'$; 
    we divide into cases based on the number of different values taken
    and apply~\cref{fact:rappor:statistics} for each subcase. Note
    that the total number of terms is
    $\ns+14\binom{\ns}{2}+36\binom{\ns}{3}+24\binom{\ns}{4}=\ns^4$. 
This can be simplified to
    \begin{align}
        \expect{N_x^2 N_y^2}
       &= \ns^2(\ns-1)^2 \mu_x^2 \mu_y^2 
       + \ns^2 (\ns-1) \mu_x \mu_y (\mu_x + \mu_y) 
       + \ns^2 \mu_x \mu_y \notag\\
&\qquad- 4\ns(\ns-1)^2 (\mu_x-\betrappor)(\mu_y-\betrappor) \mu_x \mu_y \notag\\
&\qquad- 2\ns (\ns-1)  (\mu_x-\betrappor)(\mu_y-\betrappor)\Paren{ \mu_x + \mu_y } \notag\\
&\qquad+ 2\ns(\ns-1) (\mu_x-\betrappor)^2(\mu_y-\betrappor)^2-\ns (\mu_x-\betrappor)(\mu_y-\betrappor)\,,  \label{eq:identity:rappor:cov:22}
    \end{align}
Plugging the bounds from~\cref{eq:identity:rappor:cov:22,eq:identity:rappor:cov:11,eq:identity:rappor:cov:21} into~\cref{eq:identity:rappor:cov:0} and simplifying, we get
\begin{align*}
      &\quad\cov( g(N_x,\lambda_x), g(N_y,\lambda_y) ) \\
      &\leq  2\ns(\ns-1)(\mu_x-\betrappor)(\mu_y-\betrappor)\Paren{ (\mu_x-\betrappor)(\mu_y-\betrappor) - 2(\ns-1)(\mu_x-\lambda_x)(\mu_y-\lambda_y) } \\
      &=  2\alprappor^4\ns(\ns-1)\p(x)\p(y)\Paren{ \p(x)\p(y) - 2(\ns-1)(\p(x)-\q(x))(\p(y)-\q(y)) }
\end{align*}
Summing over all distinct $x,y$, we have
$
    \sum_{1\leq x\neq y\leq \ab} \p(x)^2\p(y)^2 = \normtwo{\p}^2-\norm{\p}_4^4 \leq \normtwo{\p}^2
$
and
\[
  -\sum_{1\leq x\neq y\leq \ab} \p(x)\p(y)(\p(x)-\q(x))(\p(y)-\q(y))
   =\sum_{x\in[\ab]} \p(x)^2(\p(x)-\q(x))^2 - \Big({\sum_{x\in[\ab]} \p(x)(\p(x)-\q(x))}\Big)^2,
\]
which is at most $\sum_{x\in[\ab]} \p(x)^2(\p(x)-\q(x))^2 \leq  \normtwo{\p-\q}^2$. Thus, 
\begin{equation}
  \label{eq:identity:rappor:variance:cov}
  2\sum_{x<y} \cov( g(N_x,\lambda_x), g(N_y,\lambda_y) ) \leq 2\alprappor^4\ns^2\normtwo{\p}^2 + 4\alprappor^4\ns^3 \normtwo{\p-\q}^2\,,
\end{equation}
completing our bound for the cross-variance terms. Combining~\cref{eq:identity:rappor:variance:var,eq:identity:rappor:variance:cov} into~\cref{eq:identity:rappor:variance:all} lets us conclude that
    \[
        \var[T] \leq \ns^2\Paren{\frac{1}{8}\ab+2\alprappor^4\normtwo{\p}^2}  + \alprappor^2\ns^3\normtwo{\p-\q}^2\Paren{1+4\alprappor^2} 
        \leq 2\ab\ns^2 + 5\alprappor^2\ns^3\normtwo{\p-\q}^2 \,,
    \]
    which holds as long as $k \ge 2$, proving the lemma.
\end{proof}

\section*{Proof of~\cref{theorem:random:product:subsets}}\label{app:concentration:bivariate}

\begin{theorem}[Joint probability perturbation concentration, restated]
Consider a matrix $\delta\in\R^{\ab\times \ab}$ such that, for every
$i_0,j_0\in[\ab]$, $\sum_{j\in [\ab]}\delta_{i_0,j}=\sum_{i\in
  [\ab]}\delta_{i,j_0} = 0$. Let random variables $X=(X_1, \dots,
X_\ab)$ and $Y=(Y_1, \dots, Y_\ab)$ be independent and uniformly
distributed over length-$\ab$ binary sequences.
Define $Z = \sum_{(i,j)\in [\ab]\times[\ab]}\delta_{ij}X_iY_j$.
Then, for every $\alpha\in(0,1/16)$,
there exists a constant $c_\alpha>0$ such that
\[
\probaOf{Z^2\geq \alpha\norm{\delta}_F^2}\geq c_\alpha.
\]
Moreover, one can take $c_\alpha = \frac{(1-16\alpha)^2}{1024}$.
\end{theorem}
\begin{proof}
The proof is similar in flavor to that of~\cite[Theorem~A.6]{ACT:19}  (for the case $L=2$), as we proceed by
bounding $\expect{Z}$, $\expect{Z^2}$, and $\expect{Z^4}$, before applying the Paley--Zygmund inequality to $Z^2$. While we could follow the approach of~\cite[Theorem~A.6]{ACT:19} and handle general 4-symmetric random variables by carefully keeping track of the various 
quantities in the expansion of $\expect{Z^2}$ and $\expect{Z^4}$, for conciseness we choose here to provide a simpler (albeit less general) proof relying on our specific choice of random variables.

As a first step, let $\theta_i\eqdef 2X_i-1$ and $\theta_j'\eqdef 2Y_j-1$ for $i,j\in[\ab]$, so that the $\theta_i$ and $\theta_j'$ are independent Rademacher random variables. 
Since the sum of entries of $\delta$ along any fixed row or column is zero by assumption, we note that
\begin{align}
\label{eqn:z-normalized}
Z =  \frac{1}{4} \sum_{i,j\in[\ab]} \delta_{ij} \theta_i \theta'_j.
\end{align}
Since $\theta_i$ and $\theta_j'$ are independent and $\expect{\theta_i}=0$, it follows that $\expect{Z}=0$. 
For $Z^2$, we again use independence of $\theta$ and $\theta'$ to obtain
\begin{align*}
\expect{Z^2} 
&=  \sum_{(i_1,j_1, i_2,j_2)\in[\ab]^4}\delta_{i_1j_1}\delta_{i_2j_2} \expect{\theta_{i_1}\theta_{i_2}\theta'_{j_1}\theta'_{j_2}}
=\sum_{(i_1,j_1, i_2,j_2)\in[\ab]^4}\delta_{i_1j_1}\delta_{i_2j_2} \expect{\theta_{i_1}\theta_{i_2}}\expect{\theta'_{j_1}\theta'_{j_2}}.
 \end{align*}
Moreover, since the coordinates are independent, we also have $\expect{\theta_{i_1}\theta_{i_2}}= \expect{\theta'_{i_1}\theta'_{i_2}}=\indic{i_1=i_2}$. Therefore, 
\begin{align*}
\expect{Z^2} 
=\frac1{16} \sum_{(i,j)\in[\ab]^2}\delta_{ij}^2 = \frac{1}{16}\norm{\delta}_F^2\,.
 \end{align*}

It remains to bound the fourth moment of $Z$. Using the representation of $Z$ as in~\eqref{eqn:z-normalized}, we bound the moment-generating function of $Z$ as\footnote{
See,~\eg,~\cite[Claim~IV.17]{AcharyaCT:IT1}, and note that the proof goes through even without the positive semi-definiteness assumption.}
\[
  \log\bE{\theta \theta'}{e^{\lambda Z}}
  = \log\bE{\theta \theta'}{e^{\frac{\lambda}{4} \theta^T \delta \theta'}} \leq \frac{\lambda^2}{32} \cdot \frac{\norm{\delta}_F^2}{1-\frac{\lambda^2}{4} {\rho(\delta^T\delta)}}, \qquad \forall\, 0<\lambda<\frac{2}{\sqrt{\rho(\delta^T\delta)}},
\]
where $\rho(\delta^T\delta)$ is the spectral radius of $\delta^T \delta$. Now, by a standard Markov-based argument, we have that $\bEE{Z^4} \leq \frac{4!}{\lambda^4}\bE{\theta \theta'}{e^{\lambda Z}}$ for all $\lambda > 0$. Therefore, combining the two and using the fact that $\sqrt{\rho(\delta^T\delta)} \leq \norm{\delta}_F$ we can write
\[
    \bEE{Z^4} \leq \frac{24}{\lambda^4}e^{\frac{\lambda^2}{32} \cdot \frac{\norm{\delta}_F^2}{1-\lambda^2 \norm{\delta}_F^2/4}}, \qquad \forall\, 0<\lambda<\frac{2}{\norm{\delta}_F}.
\]
Setting $\lambda = \frac{1}{C\norm{\delta}_F}$ for any constant $C>0$ yields
$
    \bEE{Z^4} \leq 24\cdot C^4 e^{\frac{1}{32(C^2-1/4)}} \norm{\delta}_F^4\,.
$
Optimizing for $C>1/2$, we can take $C=\frac{1+\sqrt{65}}{16}$ and get 
\begin{equation}
  \bEE{Z^4} \leq 4 \norm{\delta}_F^4.
\end{equation}
The remainder of the proof follows that of~\cref{theo:random:subset} using the Paley--Zygmund inequality: for every $t\in[0,1]$
\[
    \probaOf{Z^2 > \frac{t}{16}\norm{\delta}_F^2} \geq (1-t)^2 \frac{\bEE{Z^2}^2}{\bEE{Z^4}}
    \geq \frac{(1-t)^2}{256\cdot 4},
\]
establishing the theorem (by choosing $t\eqdef 16\alpha$ and $c_\alpha\eqdef \frac{(1-16\alpha)^2}{1024}$).
\end{proof}
\begin{remark}
  Although the proof of~\cref{theorem:random:product:subsets} uses full independence of the vectors $X$ and $Y$ (due to the use of the moment-generating function), it is easy to see that the statement still holds when $X$ (resp. $Y$) is only $4$-wise independent. This is because the Paley--Zygmund-based argument only relies on bounds on moments up to order four, and those moments are the same for $4$-wise and fully independent vectors.
\end{remark}

\clearpage
  \bibliographystyle{IEEEtranS}
  \bibliography{references} 

% Generated by IEEEtranS.bst, version: 1.14 (2015/08/26)
\begin{thebibliography}{10}
\providecommand{\url}[1]{#1}
\csname url@samestyle\endcsname
\providecommand{\newblock}{\relax}
\providecommand{\bibinfo}[2]{#2}
\providecommand{\BIBentrySTDinterwordspacing}{\spaceskip=0pt\relax}
\providecommand{\BIBentryALTinterwordstretchfactor}{4}
\providecommand{\BIBentryALTinterwordspacing}{\spaceskip=\fontdimen2\font plus
\BIBentryALTinterwordstretchfactor\fontdimen3\font minus
  \fontdimen4\font\relax}
\providecommand{\BIBforeignlanguage}[2]{{%
\expandafter\ifx\csname l@#1\endcsname\relax
\typeout{** WARNING: IEEEtranS.bst: No hyphenation pattern has been}%
\typeout{** loaded for the language `#1'. Using the pattern for}%
\typeout{** the default language instead.}%
\else
\language=\csname l@#1\endcsname
\fi
#2}}
\providecommand{\BIBdecl}{\relax}
\BIBdecl

\bibitem{ACFT:19}
\BIBentryALTinterwordspacing
J.~Acharya, C.~Canonne, C.~Freitag, and H.~Tyagi, ``Test without trust: Optimal
  locally private distribution testing,'' in \emph{Proceedings of Machine
  Learning Research}, ser. Proceedings of Machine Learning Research,
  K.~Chaudhuri and M.~Sugiyama, Eds., vol.~89.\hskip 1em plus 0.5em minus
  0.4em\relax PMLR, 16--18 Apr 2019, pp. 2067--2076. [Online]. Available:
  \url{http://proceedings.mlr.press/v89/acharya19b.html}
\BIBentrySTDinterwordspacing

\bibitem{AcharyaCT:IT1}
J.~Acharya, C.~L. Canonne, and H.~Tyagi, ``Inference under information
  constraints {I}: Lower bounds from chi-square contraction,'' 2018, in
  submission. Preprint available at arXiv:abs/1812.11476.

\bibitem{ACT:19}
------, ``Inference under information constraints {II}: Communication
  constraints and shared randomness,'' 2019, in submission. Preprint available
  at arXiv:abs/1804.06952.

\bibitem{ACHST:20}
\BIBentryALTinterwordspacing
J.~Acharya, C.~L. Canonne, Y.~Han, Z.~Sun, and H.~Tyagi, ``Domain compression
  and its application to randomness-optimal distributed goodness-of-fit,'' in
  \emph{Proceedings of Thirty Third Conference on Learning Theory}, ser.
  Proceedings of Machine Learning Research, J.~Abernethy and S.~Agarwal, Eds.,
  vol. 125.\hskip 1em plus 0.5em minus 0.4em\relax PMLR, 09--12 Jul 2020, pp.
  3--40. [Online]. Available:
  \url{http://proceedings.mlr.press/v125/acharya20a.html}
\BIBentrySTDinterwordspacing

\bibitem{ACLST:20}
J.~Acharya, C.~L. Canonne, Y.~Liu, Z.~Sun, and H.~Tyagi, ``Interactive
  inference under information constraints,'' 2020.

\bibitem{AcharyaDK15}
J.~Acharya, C.~Daskalakis, and G.~C. Kamath, ``{Optimal Testing for Properties
  of Distributions},'' in \emph{Advances in Neural Information Processing
  Systems 28}, C.~Cortes, N.~Lawrence, D.~Lee, M.~Sugiyama, R.~Garnett, and
  R.~Garnett, Eds.\hskip 1em plus 0.5em minus 0.4em\relax Curran Associates,
  Inc., 2015, pp. 3577--3598.

\bibitem{AcharyaS:19}
\BIBentryALTinterwordspacing
J.~Acharya and Z.~Sun, ``Communication complexity in locally private
  distribution estimation and heavy hitters,'' in \emph{Proceedings of the 36th
  International Conference on Machine Learning}, ser. Proceedings of Machine
  Learning Research, K.~Chaudhuri and R.~Salakhutdinov, Eds., vol.~97.\hskip
  1em plus 0.5em minus 0.4em\relax Long Beach, California, USA: PMLR, 09--15
  Jun 2019, pp. 51--60. [Online]. Available:
  \url{http://proceedings.mlr.press/v97/acharya19c.html}
\BIBentrySTDinterwordspacing

\bibitem{ASZ:18:DP}
\BIBentryALTinterwordspacing
J.~Acharya, Z.~Sun, and H.~Zhang, ``Differentially private testing of identity
  and closeness of discrete distributions,'' in \emph{Advances in Neural
  Information Processing Systems 31}, S.~Bengio, H.~Wallach, H.~Larochelle,
  K.~Grauman, N.~Cesa-Bianchi, and R.~Garnett, Eds.\hskip 1em plus 0.5em minus
  0.4em\relax Curran Associates, Inc., 2018, pp. 6878--6891. [Online].
  Available:
  \url{http://papers.nips.cc/paper/7920-differentially-private-testing-of-identity-and-closeness-of-discrete-distributions.pdf}
\BIBentrySTDinterwordspacing

\bibitem{ASZ:18:HR}
\BIBentryALTinterwordspacing
------, ``Hadamard response: Estimating distributions privately, efficiently,
  and with little communication,'' in \emph{Proceedings of Machine Learning
  Research}, ser. Proceedings of Machine Learning Research, K.~Chaudhuri and
  M.~Sugiyama, Eds., vol.~89.\hskip 1em plus 0.5em minus 0.4em\relax PMLR,
  16--18 Apr 2019, pp. 1120--1129. [Online]. Available:
  \url{http://proceedings.mlr.press/v89/acharya19a.html}
\BIBentrySTDinterwordspacing

\bibitem{ADKR:19}
\BIBentryALTinterwordspacing
M.~Aliakbarpour, I.~Diakonikolas, D.~Kane, and R.~Rubinfeld, ``Private testing
  of distributions via sample permutations,'' in \emph{Advances in Neural
  Information Processing Systems 32}, H.~Wallach, H.~Larochelle,
  A.~Beygelzimer, F.~d\textquotesingle Alch\'{e}-Buc, E.~Fox, and R.~Garnett,
  Eds.\hskip 1em plus 0.5em minus 0.4em\relax Curran Associates, Inc., 2019,
  pp. 10\,878--10\,889. [Online]. Available:
  \url{http://papers.nips.cc/paper/9270-private-testing-of-distributions-via-sample-permutations.pdf}
\BIBentrySTDinterwordspacing

\bibitem{ADR:17}
\BIBentryALTinterwordspacing
M.~Aliakbarpour, I.~Diakonikolas, and R.~Rubinfeld, ``Differentially private
  identity and equivalence testing of discrete distributions,'' in
  \emph{Proceedings of the 35th International Conference on Machine Learning},
  ser. Proceedings of Machine Learning Research, J.~Dy and A.~Krause, Eds.,
  vol.~80.\hskip 1em plus 0.5em minus 0.4em\relax Stockholmsmässan, Stockholm
  Sweden: PMLR, 10--15 Jul 2018, pp. 169--178. [Online]. Available:
  \url{http://proceedings.mlr.press/v80/aliakbarpour18a.html}
\BIBentrySTDinterwordspacing

\bibitem{AJM:20}
\BIBentryALTinterwordspacing
K.~Amin, M.~Joseph, and J.~Mao, ``Pan-private uniformity testing,'' in
  \emph{Proceedings of Thirty Third Conference on Learning Theory}, ser.
  Proceedings of Machine Learning Research, J.~Abernethy and S.~Agarwal, Eds.,
  vol. 125.\hskip 1em plus 0.5em minus 0.4em\relax PMLR, 09--12 Jul 2020, pp.
  183--218. [Online]. Available:
  \url{http://proceedings.mlr.press/v125/amin20a.html}
\BIBentrySTDinterwordspacing

\bibitem{BW:17}
\BIBentryALTinterwordspacing
S.~Balakrishnan and L.~Wasserman, ``Hypothesis testing for high-dimensional
  multinomials: {A} selective review,'' \emph{The Annals of Applied
  Statistics}, vol.~12, no.~2, pp. 727--749, 2018. [Online]. Available:
  \url{https://doi.org/10.1214/18-AOAS1155SF}
\BIBentrySTDinterwordspacing

\bibitem{BCJM:20}
\BIBentryALTinterwordspacing
V.~Balcer, A.~Cheu, M.~Joseph, and J.~Mao, ``Connecting robust shuffle privacy
  and pan-privacy,'' \emph{CoRR}, vol. abs/2004.09481, 2020. [Online].
  Available: \url{https://arxiv.org/abs/2004.09481}
\BIBentrySTDinterwordspacing

\bibitem{BNST:17}
\BIBentryALTinterwordspacing
R.~Bassily, K.~Nissim, U.~Stemmer, and A.~Guha~Thakurta, ``Practical locally
  private heavy hitters,'' in \emph{Advances in Neural Information Processing
  Systems}, I.~Guyon, U.~V. Luxburg, S.~Bengio, H.~Wallach, R.~Fergus,
  S.~Vishwanathan, and R.~Garnett, Eds., vol.~30.\hskip 1em plus 0.5em minus
  0.4em\relax Curran Associates, Inc., 2017, pp. 2288--2296. [Online].
  Available:
  \url{https://proceedings.neurips.cc/paper/2017/file/3d779cae2d46cf6a8a99a35ba4167977-Paper.pdf}
\BIBentrySTDinterwordspacing

\bibitem{BFFKRW:01}
T.~Batu, E.~Fischer, L.~Fortnow, R.~Kumar, R.~Rubinfeld, and P.~White,
  ``Testing random variables for independence and identity,'' in \emph{42nd
  Annual Symposium on Foundations of Computer Science, {FOCS} 2001}, 2001, pp.
  442--451.

\bibitem{BB:20}
T.~B. Berrett and C.~Butucea, ``Locally private non-asymptotic testing of
  discrete distributions is faster using interactive mechanisms,'' \emph{CoRR},
  vol. abs/2005.12601, 2020.

\bibitem{BCG:17}
E.~Blais, C.~L. Canonne, and T.~Gur, ``Distribution testing lower bounds via
  reductions from communication complexity,'' in \emph{32nd {C}omputational
  {C}omplexity {C}onference}, ser. LIPIcs. Leibniz Int. Proc. Inform.\hskip 1em
  plus 0.5em minus 0.4em\relax Schloss Dagstuhl. Leibniz-Zent. Inform., Wadern,
  2017, vol.~79, pp. Art. No. 28, 40.

\bibitem{BCG:19}
\BIBentryALTinterwordspacing
------, ``Distribution testing lower bounds via reductions from communication
  complexity,'' \emph{ACM Trans. Comput. Theory}, vol.~11, no.~2, pp. Art. 6,
  37, 2019, journal version of~\cite{BCG:17}. [Online]. Available:
  \url{https://doi.org/10.1145/3305270}
\BIBentrySTDinterwordspacing

\bibitem{CDK:17}
B.~Cai, C.~Daskalakis, and G.~Kamath, ``Priv'it: Private and sample efficient
  identity testing,'' in \emph{Proceedings of the 34th International Conference
  on Machine Learning}, ser. ICML '17.\hskip 1em plus 0.5em minus 0.4em\relax
  JMLR, Inc., 2017, pp. 635--644.

\bibitem{Canonne:15}
\BIBentryALTinterwordspacing
C.~L. Canonne, ``{Big Data on the Rise? Testing Monotonicity of
  Distributions},'' in \emph{Proceedings of ICALP}.\hskip 1em plus 0.5em minus
  0.4em\relax Springer, 2015, pp. 294--305. [Online]. Available:
  \url{http://dx.doi.org/10.1007/978-3-662-47672-7_24}
\BIBentrySTDinterwordspacing

\bibitem{CDKS:17}
C.~L. Canonne, I.~Diakonikolas, D.~M. Kane, and A.~Stewart, ``{Testing Bayesian
  Networks},'' in \emph{Proceedings of the 2017 Conference on Learning Theory},
  ser. Proceedings of Machine Learning Research, S.~Kale and O.~Shamir, Eds.,
  vol.~65.\hskip 1em plus 0.5em minus 0.4em\relax Amsterdam, Netherlands: PMLR,
  07--10 Jul 2017, pp. 370--448.

\bibitem{CKMUZ:19}
C.~L. Canonne, G.~Kamath, A.~McMillan, J.~Ullman, and L.~Zakynthinou, ``Private
  identity testing for high-dimensional distributions,'' \emph{CoRR}, vol.
  abs/1905.11947, 2019.

\bibitem{ChanDVV14}
S.~Chan, I.~Diakonikolas, G.~Valiant, and P.~Valiant, ``Optimal algorithms for
  testing closeness of discrete distributions,'' in \emph{Proceedings of SODA},
  2014, pp. 1193--1203.

\bibitem{DGPP:18}
I.~Diakonikolas, T.~Gouleakis, J.~Peebles, and E.~Price, ``Sample-optimal
  identity testing with high probability,'' in \emph{{ICALP}}, ser. LIPIcs,
  vol. 107.\hskip 1em plus 0.5em minus 0.4em\relax Schloss Dagstuhl -
  Leibniz-Zentrum fuer Informatik, 2018, pp. 41:1--41:14.

\bibitem{DK:16}
I.~{Diakonikolas} and D.~M. {Kane}, ``A new approach for testing properties of
  discrete distributions,'' in \emph{57th Annual {IEEE} Symposium on
  Foundations of Computer Science, {FOCS} 2016}.\hskip 1em plus 0.5em minus
  0.4em\relax {IEEE} Computer Society, 2016.

\bibitem{DJW:13}
J.~C. Duchi, M.~I. Jordan, and M.~J. Wainwright, ``Local privacy and
  statistical minimax rates,'' in \emph{54th Annual {IEEE} Symposium on
  Foundations of Computer Science, {FOCS} 2013}.\hskip 1em plus 0.5em minus
  0.4em\relax {IEEE} Computer Society, 2013, pp. 429--438.

\bibitem{DMNS:06}
C.~Dwork, F.~McSherry, K.~Nissim, and A.~Smith, ``Calibrating noise to
  sensitivity in private data analysis,'' in \emph{Theory of cryptography},
  ser. Lecture Notes in Comput. Sci.\hskip 1em plus 0.5em minus 0.4em\relax
  Springer, Berlin, 2006, vol. 3876, pp. 265--284.

\bibitem{ErlingssonPK14}
{\'U}.~Erlingsson, V.~Pihur, and A.~Korolova, ``{RAPPOR}: Randomized
  aggregatable privacy-preserving ordinal response,'' in \emph{Proceedings of
  the 2014 ACM Conference on Computer and Communications Security}, ser. CCS
  '14.\hskip 1em plus 0.5em minus 0.4em\relax New York, NY, USA: ACM, 2014, pp.
  1054--1067.

\bibitem{EvfimievskiGS:03}
A.~V. Evfimievski, J.~Gehrke, and R.~Srikant, ``Limiting privacy breaches in
  privacy preserving data mining,'' in \emph{{PODS}}.\hskip 1em plus 0.5em
  minus 0.4em\relax {ACM}, 2003, pp. 211--222.

\bibitem{GaboardiLRV:16}
M.~Gaboardi, H.~Lim, R.~M. Rogers, and S.~P. Vadhan, ``Differentially private
  chi-squared hypothesis testing: Goodness of fit and independence testing,''
  in \emph{Proceedings of the 33rd International Conference on Machine
  Learning}, ser. ICML '16.\hskip 1em plus 0.5em minus 0.4em\relax JMLR, Inc.,
  2016, pp. 1395--1403.

\bibitem{GR:18}
\BIBentryALTinterwordspacing
M.~Gaboardi and R.~Rogers, ``Local private hypothesis testing: {C}hi-square
  tests,'' in \emph{Proceedings of the 35th International Conference on Machine
  Learning}, ser. Proceedings of Machine Learning Research, J.~Dy and
  A.~Krause, Eds., vol.~80.\hskip 1em plus 0.5em minus 0.4em\relax
  Stockholmsmässan, Stockholm Sweden: PMLR, 10--15 Jul 2018, pp. 1626--1635.
  [Online]. Available: \url{http://proceedings.mlr.press/v80/gaboardi18a.html}
\BIBentrySTDinterwordspacing

\bibitem{Goldreich:16}
\BIBentryALTinterwordspacing
O.~Goldreich, ``The uniform distribution is complete with respect to testing
  identity to a fixed distribution,'' \emph{Electronic Colloquium on
  Computational Complexity (ECCC)}, vol.~23, p.~15, 2016. [Online]. Available:
  \url{http://eccc.hpi-web.de/report/2016/015}
\BIBentrySTDinterwordspacing

\bibitem{Goldreich:17}
\BIBentryALTinterwordspacing
------, \emph{{Introduction to Property Testing}}.\hskip 1em plus 0.5em minus
  0.4em\relax Cambridge University Press, 2017. [Online]. Available:
  \url{http://www.wisdom.weizmann.ac.il/~oded/pt-intro.html}
\BIBentrySTDinterwordspacing

\bibitem{HuangM13}
D.~Huang and S.~Meyn, ``Generalized error exponents for small sample universal
  hypothesis testing,'' \emph{IEEE Transactions on Information Theory},
  vol.~59, no.~12, pp. 8157--8181, 2013.

\bibitem{JMNR:19}
M.~Joseph, J.~Mao, S.~Neel, and A.~Roth, ``The role of interactivity in local
  differential privacy,'' in \emph{{FOCS}}.\hskip 1em plus 0.5em minus
  0.4em\relax {IEEE} Computer Society, 2019, pp. 94--105.

\bibitem{KairouzBR16}
P.~Kairouz, K.~Bonawitz, and D.~Ramage, ``Discrete distribution estimation
  under local privacy,'' in \emph{Proceedings of the 33rd International
  Conference on Machine Learning, {ICML} 2016}, ser. {JMLR} Workshop and
  Conference Proceedings, vol.~48.\hskip 1em plus 0.5em minus 0.4em\relax
  JMLR.org, 2016, pp. 2436--2444.

\bibitem{KLNRS:08}
S.~P. Kasiviswanathan, H.~K. Lee, K.~Nissim, S.~Raskhodnikova, and A.~Smith,
  ``What can we learn privately?'' in \emph{49th Annual {IEEE} Symposium on
  Foundations of Computer Science, {FOCS} 2008}.\hskip 1em plus 0.5em minus
  0.4em\relax {IEEE}, Oct. 2008, pp. 531--540.

\bibitem{KLNRS:11}
------, ``What can we learn privately?'' \emph{SIAM J. Comput.}, vol.~40,
  no.~3, pp. 793--826, 2011.

\bibitem{KiferR:17}
D.~Kifer and R.~M. Rogers, ``A new class of private chi-square tests,'' in
  \emph{Proceedings of the 20th International Conference on Artificial
  Intelligence and Statistics}, ser. AISTATS '17.\hskip 1em plus 0.5em minus
  0.4em\relax JMLR, Inc., 2017, pp. 991--1000.

\bibitem{LRR:13}
\BIBentryALTinterwordspacing
R.~Levi, D.~Ron, and R.~Rubinfeld, ``Testing properties of collections of
  distributions,'' \emph{Theory of Computing}, vol.~9, pp. 295--347, 2013.
  [Online]. Available: \url{http://dx.doi.org/10.4086/toc.2013.v009a008}
\BIBentrySTDinterwordspacing

\bibitem{Paninski:08}
L.~Paninski, ``A coincidence-based test for uniformity given very sparsely
  sampled discrete data,'' \emph{IEEE Transactions on Information Theory},
  vol.~54, no.~10, pp. 4750--4755, 2008.

\bibitem{Rubinfeld:12}
\BIBentryALTinterwordspacing
R.~Rubinfeld, ``Taming big probability distributions,'' \emph{{XRDS}:
  Crossroads, The {ACM} Magazine for Students}, vol.~19, no.~1, p.~24, sep
  2012. [Online]. Available: \url{http://dx.doi.org/10.1145/2331042.2331052}
\BIBentrySTDinterwordspacing

\bibitem{Sheffet:18}
O.~Sheffet, ``Locally private hypothesis testing,'' in \emph{Proceedings of the
  35th International Conference on Machine Learning}, ser. Proceedings of
  Machine Learning Research, J.~Dy and A.~Krause, Eds., vol.~80.\hskip 1em plus
  0.5em minus 0.4em\relax Stockholmsmässan, Stockholm Sweden: PMLR, 10--15 Jul
  2018, pp. 4612--4621.

\bibitem{Tsitsiklis:93}
J.~N. Tsitsiklis, ``Decentralized detection,'' in \emph{Advances in Statistical
  Signal Processing}, H.~V. Poor and J.~B. Thomas, Eds., vol.~2.\hskip 1em plus
  0.5em minus 0.4em\relax JAI Press, 1993, pp. 297--344.

\bibitem{ValiantV17a}
G.~Valiant and P.~Valiant, ``An automatic inequality prover and instance
  optimal identity testing,'' \emph{SIAM Journal on Computing}, vol.~46, no.~1,
  pp. 429--455, 2017.

\bibitem{WangLK:15}
Y.~Wang, J.~Lee, and D.~Kifer, ``Revisiting differentially private hypothesis
  tests for categorical data,'' \emph{arXiv preprint arXiv:1511.03376}, 2015.

\bibitem{Warner:65}
S.~L. Warner, ``Randomized response: A survey technique for eliminating evasive
  answer bias,'' \emph{Journal of the American Statistical Association},
  vol.~60, no. 309, pp. 63--69, 1965.

\bibitem{YeB17}
\BIBentryALTinterwordspacing
M.~Ye and A.~Barg, ``Optimal schemes for discrete distribution estimation under
  locally differential privacy,'' \emph{IEEE Trans. Inform. Theory}, vol.~64,
  no.~8, pp. 5662--5676, 2018. [Online]. Available:
  \url{https://doi.org/10.1109/TIT.2018.2809790}
\BIBentrySTDinterwordspacing

\end{thebibliography}

\end{document}